\documentclass[final, twocolum, a4paper]{IEEEtran}

\usepackage[cmex10]{amsmath}
\newcommand\numberthis{\addtocounter{equation}{1}\tag{\theequation}}	
\usepackage{amssymb}
\usepackage{amsthm}
\usepackage{graphicx}
\usepackage{pgfplots}		
\usepackage{multirow,bigdelim}
\usepackage{hyperref}
\usepackage{cancel}
\usepackage{booktabs}

\usepackage{mathtools}

\newcommand{\CC}{{\mathbb{C}}}

\newcommand{\BLACK}{\color[rgb]{0,0,0}}

\newcommand{\fv}{{\bf f}}
\newcommand{\gv}{{\bf g}}
\newcommand{\hv}{{\bf h}}

\newcommand{\sv}{{\bf s}}

\newcommand{\xv}{{\bf x}}
\newcommand{\yv}{{\bf y}}

\newcommand{\zerov}{{\bf 0}}

\newcommand{\Am}{{\bf A}}
\newcommand{\Bm}{{\bf B}}
\newcommand{\Cm}{{\bf C}}
\newcommand{\Dm}{{\bf D}}

\newcommand{\Fm}{{\bf F}}
\newcommand{\Gm}{{\bf G}}
\newcommand{\Hm}{{\bf H}}
\newcommand{\Id}{{\bf I}}

\newcommand{\Mm}{{\bf M}}

\newcommand{\Pm}{{\bf P}}
\newcommand{\Qm}{{\bf Q}}
\newcommand{\Rm}{{\bf R}}

\newcommand{\Tm}{{\bf T}}

\newcommand{\Xm}{{\bf X}}

\newcommand{\Zm}{{\bf Z}}

\newcommand{\Psim}{\hbox{\boldmath$\Psi$}}

\newcommand{\supp}{{\hbox{Supp}\,}}

\newcommand{\diag}{{\hbox{diag}}}

\newcommand{\trace}{{\hbox{\rm tr}\,}}
\newcommand{\rank}{{\hbox{rank}}}

\newtheorem{theorem}{Theorem}
\newtheorem{definition}{Definition}
\newtheorem{lemma}{Lemma}
\newtheorem{cor}{Corollary}

\newtheorem{remark}{Remark} 
\newtheorem*{remark*}{Remark}
\newtheorem{assumption}{Assumption~A-\kern-0pt} 

\DeclareMathOperator{\tr}{tr}

\newcommand{\eqtextt}[1]{\ensuremath{\stackrel{\smash{\scriptscriptstyle{\text{#1}}}}{=}}}

\newcommand\independent{\protect\mathpalette{\protect\independenT}{\perp}}
\def\independenT#1#2{\mathrel{\rlap{$#1#2$}\mkern2mu{#1#2}}}

\newcommand{\CCC}{{\mathcal C}}
\newcommand{\NNN}{{\mathcal N}}

\newcommand{\betato}{\overset{\rm \beta\to\infty}{\longrightarrow}}
\newcommand{\alphato}{\overset{\rm \alpha\to\infty}{\longrightarrow}}
\newcommand{\betaxto}{\overset{\rm \beta_x\to\infty}{\longrightarrow}}
\newcommand{\alphaxto}{\overset{\rm \alpha_x\to\infty}{\longrightarrow}}
\newcommand{\trans}{{\sf T}}
\newcommand{\eqdef}{{\overset{\Delta}{=}}}
\renewcommand{\H}{{\sf H}}
\newcommand{\mo}{{-1}}
\newcommand{\mt}{{-2}}

\newcommand{\OH}{{\frac{1}{2}}}

\newcommand{\lr}{\left( }
\newcommand{\rr}{\right) }
\newcommand{\ls}{\left[ }
\newcommand{\rs}{\right] }
\newcommand{\lc}{\left\{ }
\newcommand{\rc}{\right\} }

\newcommand{\expect}{{\mathbb E}}

\title{Interference-Aware RZF Precoding for \\ Multi Cell Downlink Systems}
\author{Axel M\"uller,
	Romain Couillet,
	Emil~Bj\"ornson,
	Sebastian Wagner,
	and~M\'erouane~Debbah
	\thanks{A.~M\"uller was with Intel Mobile Communications, Sophia Antipolis, France and with the Alcatel-Lucent Chair on Flexible Radio, Sup\'elec, Gif-sur-Yvette, France. He is currently with the Mathematical and Algorithmic Sciences Lab, France Research Center, Huawei Technologies Co. Ltd., Boulogne-Billancourt, France (email: axel.mueller@huawei.com).}
	\thanks{R.~Couillet is with Laboratoire des Signaux et Systèmes (L2S, UMR CNRS 8506), CentraleSupelec - CNRS - Universit\'e Paris-Sud, Gif-sur-Yvette, France (email: romain.couillet@centralesupelec.fr).}
	\thanks{E.~Bj\"ornson was with the Alcatel-Lucent Chair on Flexible Radio, Sup\'elec, Gif-sur-Yvette, France, and with the Department of Signal Processing, KTH Royal Institute of Technology, Stockholm, Sweden. He is currently with the Department of Electrical Engineering (ISY), Link\"{o}ping University, Link\"{o}ping, Sweden (email: emil.bjornson@liu.se).}
	\thanks{S.~Wagner is with Intel Mobile Communications, Sophia Antipolis, France (email: sebastian.wagner@intel.com).}
	\thanks{M.~Debbah is with Laboratoire des Signaux et Syst\`emes (L2S, UMR CNRS 8506), CentraleSupelec - CNRS - Universit\'e Paris-Sud, Gif-sur-Yvette, France, and with the Mathematical and Algorithmic Sciences Lab, France Research Center, Huawei Technologies Co. Ltd., Boulogne-Billancourt, France (email: merouane.debbah@centralesupelec.fr).}
	\thanks{E.~Bj\"ornson was funded by the International Postdoc Grant 2012-228 from The Swedish Research Council. This research has been supported by the ERC Starting Grant 305123 MORE (Advanced Mathematical Tools for Complex Network Engineering).}
}

\begin{document}

\maketitle

\begin{abstract} 
Recently, a structure of an optimal linear precoder for multi cell downlink systems has been described in \cite[Eq~(3.33)]{Bjornson2013d}. Other references (e.g., \cite{Jose2011b,Hoydis2013c}) have used simplified versions of the precoder to obtain promising performance gains. These gains have been hypothesized to stem from the additional degrees of freedom that allow for interference mitigation through interference relegation to orthogonal subspaces. However, no conclusive or rigorous understanding has yet been developed.

In this paper, we build on an intuitive interference induction trade-off and the aforementioned precoding structure to propose an interference aware RZF (iaRZF) precoding scheme for multi cell downlink systems and we analyze its rate performance. Special emphasis is placed on the induced interference mitigation mechanism of iaRZF. For example, we will verify the intuitive expectation that the precoder structure can either completely remove induced inter-cell or intra-cell interference.
We state new results from large-scale random matrix theory that make it possible to give more intuitive and insightful explanations of the precoder behavior, also for cases involving imperfect channel state information (CSI).
We remark especially that the interference-aware precoder makes use of all available information about interfering channels to improve performance. Even very poor CSI allows for significant sum-rate gains.
Our obtained insights are then used to propose heuristic precoder parameters for arbitrary systems, whose effectiveness are shown in more involved system scenarios. Furthermore, calculation and implementation of  these parameters does not require explicit inter base station cooperation.
\end{abstract}

\section{Introduction}
\label{sec:iaRZF_Intro}

The growth of data traffic and the number of user terminals (UTs) in cellular networks will likely persist for the foreseeable future \cite{CISCO2013}. In order to deal with the resulting demand, it is estimated \cite{1000xChallenge2013} that a thousand-fold increase in network capacity is required over the next $10$ years. 
Given that the available spectral resources are severely limited, the majority of the wireless community sees massive network densification as the most realistic approach to solving most pressing issues. Also historically, shrinking cell size has been the single most successful technique in satisfying demand for network capacity \cite[Chapter~6.3.4]{Webb2007}. In recent times, this technique has been named the \emph{small cell} approach \cite{andrews2012femtocells, Hoydis2011c}. 
A large body of research indicates that interference still is a major limiting factor for capacity in multi cell scenarios \cite{Dai2004,GesbertYu2010}, especially in modern cellular networks that serve a multitude of users within the same time/frequency resources.
In general, we see a trend to using more and more antennas for interference mitigation, e.g., via the massive MIMO approach \cite{Marzetta2010a}. Here, the number of transmit antennas surpasses the number of served UTs by an order of magnitude.
Independent of this specific approach, the surplus antennas can be used to mitigate interference by using spatial precoding \cite{Bjornson2013d,Irmer2011,GesbertSalzer2007,GesbertYu2010}.
The interference problem is generally compounded by the effect of imperfect knowledge concerning the channel state information (CSI).
Such imperfections are unavoidable, as imperfect estimation algorithms, limited number of orthogonal pilot sequences, mobile UTs, delays, etc.\ can not be avoided in practice. 
Hence, one is interested in employing precoding schemes that are robust to CSI estimation errors and exploit the available CSI as efficiently as possible.

Arguably, the most successful and practically applicable precoding scheme used today is RZF precoding \cite{Peel2005a} (also known as minimum mean square error (MMSE) precoding, transmit Wiener filter, generalized eigenvalue-based beamformer, etc.; see \cite[Remark~3.2]{Bjornson2013d}).
Classical RZF precoders are only defined for single cell systems and thus do not take inter cell interference into account. Disregarding available information about inter cell interference is particularly detrimental in high density scenarios, where interference is a main performance limiting factor.
Early multi cell extensions of the RZF scheme do not take the quality of CSI into account \cite{Zakhour2010a} and later ones either rely on heuristic distributed optimization algorithms or on inter cell cooperation \cite{Bjornson2009e} to determine the precoding vector. Thus, they offer limited insight into the precoder structure, i.e., into how the precoder works and how it could be improved.
%

An intuitive extension of the single cell RZF, with the goal of completely eliminating induced interference is to substitute the intra cell channel matrix $\Hm$ in the (qualitative, single antenna UTs) precoder formulation\footnote{Where $\xi$ represents some positive regularization parameter.} $\Fm=\Hm(\Hm^\H\Hm+\xi\Id)^\mo$ by a matrix $\check{\Hm}$, which is $\Hm$ projected onto the space orthogonal to the inter cell channel matrices, i.e., $\check{\Fm}=\check{\Hm}\lr \check{\Hm}^\H\check{\Hm}+\xi\Id\rr^\mo$. Hence induced interference can be completely removed at the cost of reduced signal power, if the CSI is perfectly known. 
However, it is immediately clear that this is a very harsh requirement, since the projection negatively affects the amount of signal energy received at the served UTs (unless $\Hm=\check{\Hm}$).
Assuming the precoding objective is system wide sum-rate optimization, one realizes that single cell RZF is probably not optimal, since it reduces the rate in other cells due to induced interference. 
Thus, a trade-off between the two extremes is expected to be beneficial, especially when the channel matrices are estimated with dissimilar quality.
In this paper we analyze the following class of precoders for multi-cell single antenna UT systems, which we will denote \emph{interference-aware RZF} (iaRZF). This class allows for the desired trade-off, as will be shown later on:
\begin{align}
\Fm^m_m 
& = \lr \sum_{l=1}^L \alpha^m_l \hat{\Hm}^m_l (\hat{\Hm}^m_l)^\H + \xi_m \Id_{N_m} \rr^\mo \hat{\Hm}^m_m \nu_m^\OH \label{eq:iaRZF_iaRZFprecoder} \,.
\end{align}
Here $\Fm^m_m$ is the linear precoder used by base station (BS) $m$ and $\hat{\Hm}^m_l$ denotes the imperfect estimate of the channel matrix from BS $m$ to the UTs in cell $l$. The factor $\xi_m$ is a regularization parameter and $\nu_m$ normalizes the precoder. Each channel matrix is assigned a factor $\alpha^m_l$, that can be interpreted as the importance placed on the respective estimated channel.
It is easy to see how this structure can mimic single cell RZF under perfect CSI (choose $\alpha^m_l=0$, $l\neq m$ and $\alpha^m_m=1$).
The weights $\alpha^m_l$ allow balancing signal power directed to the served users with interference induced to other cells. This can be used to optimize sum-rate performance in certain cases, as will be shown in Section~\ref{sec:iaRZF_SimpleiaRZF}. 
In general the optimal weights are not known and the classical UL/DL duality approach (e.g., \cite{dahrouj2010coordinated}, \cite{Bjornson2013d}) cannot be applied to find these weights, outside of power minimization settings or, if imperfect CSI is considered. 
We note that every BS can try to estimate the interference from other cells without explicit inter cell cooperation or communication, by means of blind or known pilot based schemes, though the CSI quality might be rather poor. Such estimation might be considered as implicit coordination.
In \cite{Hoydis2013c} a simplified version of iaRZF was discussed, where a single subset of UT channels was weighted with respect to an estimated receive covariance matrix of all interfering channels.
Hoydis \emph{et~al.} argued that ``large [weights] make the precoding vectors more orthogonal to the interference subspace'', but they did not conclusively and rigorously show how or why this is achieved.
The work in \cite{Hoydis2013c} builds upon results from \cite[Theorem~6]{Jose2011b}, which introduced the precoding structure offering a single common balance parameter weighting all inter cell interference. 
Differing from these works, we are more interested in increasing the sum-rate performance by giving the traditional RZF precoder even more additional degrees of freedom. Doing so enables it to separately take into account the interference induced to certain subspaces, associated with different cells. Increasing the degrees of freedom in the interference suppression was also shown to be effective in \cite{Bjornson2012c} for perfect CSI.
Similar to the traditional RZF precoder our approach is empirical and based on several motivational aspects (see Section~\ref{sec:iaRZF_SimpleiaRZF}).
The iaRZF structure is also partially based on the work in \cite{dahrouj2010coordinated} and \cite[Eq~(3.33)]{Bjornson2013d}. In the latter, one finds one of the most recent and general treatments of the multi cell RZF precoder, along with proof that the proposed structure is optimal w.r.t.\ many utility functions of practical interest (see also \cite{Bjoernson2014}).

This paper analyzes the proposed iaRZF scheme, showing that it can significantly improve sum-rate performance in high interference multi cellular scenarios. In particular, it is not necessary to have reliable estimations of interfering channels; even very poor CSI allow for significant gains. 
We facilitate intuitive understanding of the precoder through new methods of analysis in both finite and large dimensions. Special emphasis is placed on the induced interference mitigation mechanism of iaRZF.
To obtain fundamental insights, we consider the large-system regime in which the number of transmit antennas and UTs are both large. 
Furthermore, new finite dimensional approaches for analyzing multi cell RZF precoding schemes are introduced and applied for limiting cases.
We derive deterministic expressions for the asymptotic user rates, which also serve as accurate approximations in practical non-asymptotic regimes. Merely the channel statistics are needed for calculation and implementation of our deterministic expressions.
These novel expressions generalize the prior work in \cite{Wagner2012a} for single cell systems and in \cite{Hoydis2013a} for multi cell systems, where only deterministic channel covariance matrices can be used in the analysis pertaining to the suppression of inter cell interference. 
Then, these extensions are used to optimize the sum-rate of the iaRZF precoding scheme in limiting cases.
Insights gathered from this lead us to propose and motivate an appropriate heuristic scaling of the precoder weights w.r.t.\ various system parameters, that offers attractive sum-rate performance; also in non-limiting cases.

The notation in this paper adheres to the following general rules.
Boldface lower case is used for column vectors, and upper case for matrices. ${\bf X}^{\mbox{\tiny T}}$ and ${\bf X}^{\mbox{\tiny H}}$  denote the transpose, and conjugate transpose of ${\bf X}$, respectively, while $\tr ({\bf X})$ is the matrix trace function.
The expectation operator is denoted $\mathbb{E}[\cdot]$.
The spectral norm of $\Xm$ is denoted $\|\Xm\|_2$ and the Euclidean norm of $\xv$ is denoted $\|\xv\|_2$. Circularly symmetric complex Gaussian random vectors are denoted $\mathcal{CN}(\bar{{\bf x}},{\bf Q})$, where $\bar{{\bf x}}$ is the mean and ${\bf Q}$ is the covariance matrix.
The set of all complex numbers is denoted by $\CC$, with $\CC^{N\times 1}$ and $\CC^{N\times M}$ being the generalizations to vectors and matrices, respectively. The $M\times M$ identity matrix  is written as ${\bf I}_M$, the zero vector of length $M$ is denoted ${\bf 0}_{M\times 1}$ and the zero matrix ${\bf 0}_M$.
Throughout this paper, superscripts generally refer to the origin (e.g., cell $m$) and subscripts generally denote the destination (e.g., cell $l$ or UT $k$ of cell $l$), when both information are needed.
We employ $ \independent $ and $ \not\independent $ to mean stochastic independence and dependence, respectively.

\section{Understanding iaRZF}
\label{sec:iaRZF_SimpleiaRZF}
In order to intuitively understand and motivate the iaRZF precoder this section first analyzes its behavior and impact in a relatively simple system.

\subsection{Simple System}
\label{ssec:iaRZF_SimpleSystem}

\begin{figure}
	\centering
	\includegraphics[width=0.3\textwidth]{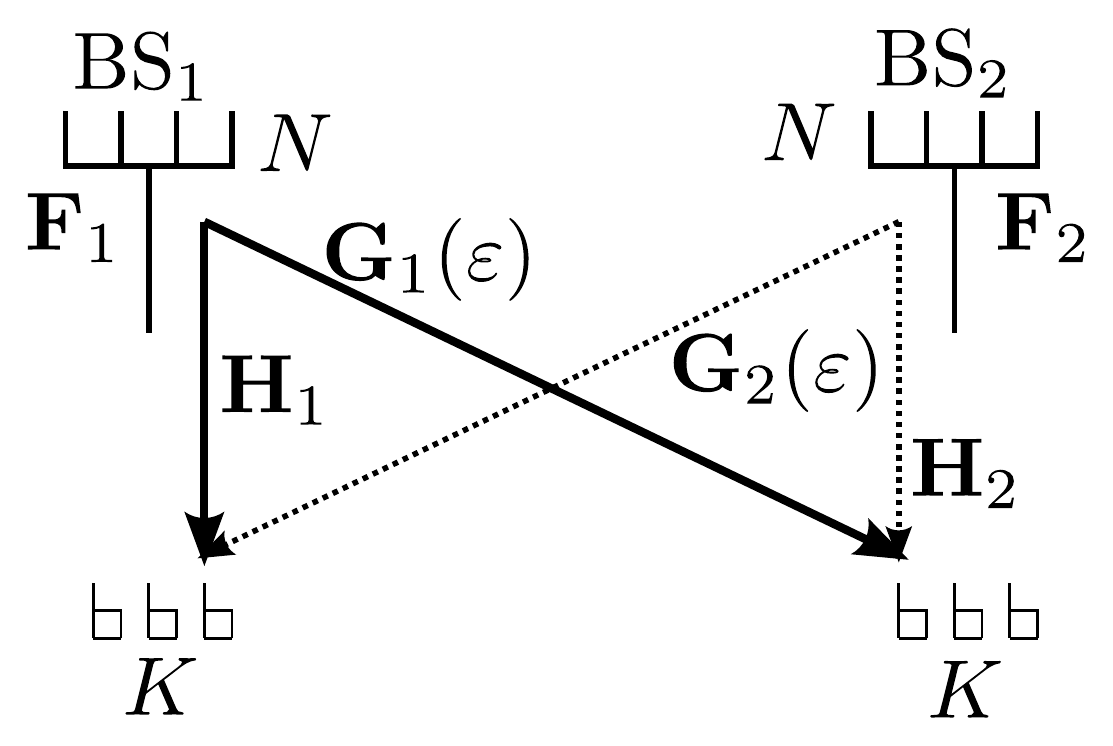}
	\caption{Simple $2$ BS downlink system.}
	\label{fig:iaRZF_Simplified2BSSystem}
\end{figure}
We start by examining a simple two cell downlink system, as depicted in Figure~\ref{fig:iaRZF_Simplified2BSSystem}, which is a further simplification of the Wyner model \cite{Wyner1994a,Xu2011a}. It features $2$ BSs, BS$_1$ and BS$_2$, with $N$ antennas each. Every BS serves one cell with $K$ single antenna users. For convenience we introduce the notations $c=K/N$ and $\bar{x}=\bmod(x,2)+1, x\in\{1, 2\}$. In order to circumvent scheduling complications, we assume $N\geq K$.
The aggregated channel matrix between BS$_x$ and the affiliated users is denoted $\Hm_x =[\hv_{x,1}, \ldots, \hv_{x,K}]\in \CC^{N\times K}$ and the matrix pertaining to the users of the other cell $\Gm_x(\varepsilon) =[\gv_{x,1}, \ldots, \gv_{x,K}]\in \CC^{N\times K}$, which is usually abbreviated as $\Gm_x$. We treat $\varepsilon$ as an arbitrary interference channel gain/path-loss factor. The precoding matrix used at BS$_x$ is denoted $\Fm_x\in\CC^{N\times K}$. 
For the channel realizations we choose a block-wise fast fading model, where $\hv_{x,k} \sim \CCC\NNN(0,\frac{1}{N}\Id_N)$ and $\gv_{x,k} \sim \CCC\NNN(0,\varepsilon\frac{1}{N}\Id_N)$ for $k=1,\ldots,K$. The scaling factor $1/N$ is introduced for technical reasons and is further explained in Remark~\ref{rem:iaRZF_1overNfactor}.

Denoting $\fv_{x,k}$ the $k$th column of $\Fm_x$, $\Fm_{x[k]}$ as $\Fm_x$ with its $k$th column removed and $n_{x,k}\sim\CCC\NNN(0,1)$ the received additive Gaussian noise at UT$_{x,k}$, we define the received signal at UT$_{x,k}$ as
\begin{align*}
y_{x,k} = \hv_{x,k}^\H \fv_{x,k} s_{x,k} + 
\underbrace{ \hv_{x,k}^\H \Fm_{x[k]} \sv_{x[k]} }_{\text{intra cell interference}} + 
\underbrace{ \gv_{\bar{x},k}^\H \Fm_{\bar{x}} \sv_{\bar{x}} }_{\text{inter cell interference}} 
+ n_{x,k}
\end{align*}
where $\sv_x \sim \CCC\NNN(0, \rho_x\Id_N)$\footnote{We remark that $\rho_x$ is of order $1$.} is the vector of transmitted Gaussian symbols. It defines the average per UT transmit power of BS$_x$ as $\rho_x$ (normalized w.r.t.\ noise). 
The notations $\sv_{x[k]}$ and $s_{x,k}$ designate the transmit vector without symbol $k$ and the transmit symbol of UT$_{x,k}$.

When calculating the precoder $\Fm_x$, we assume that the channel $\Hm_x$ can be arbitrarily well estimated, however, we allow for mis-estimation of the ``inter cell interference channel'' $\Gm_x$ by adopting the generic Gauss-Markov formulation 
\begin{align*}
\hat{\Gm}_x = \sqrt{1-\tau^2}\Gm_x + \tau \tilde{\Gm}_x \,.
\end{align*}
Choosing $\tilde{\gv}_{x,k} \sim \CCC\NNN(0,\varepsilon\frac{1}{N}\Id_N)$, we can vary the available CSI quality by adjusting $0\leq\tau\leq 1$ appropriately.

In this section we use the previously introduced iaRZF precoding scheme. Given our simple system the unnormalized precoder reads
\begin{align}
\Mm_x = \lr \alpha_x\Hm_x\Hm_x^\H+\beta_x\hat{\Gm}_x\hat{\Gm}_x^\H + \xi_x\Id  \rr^\mo \Hm_x \label{eq:iaRZF_Mm_x} \,. 
\end{align}
One remarks that the regularization of the identity matrix can also be controlled by scaling $\alpha_x$ and $\beta_x$ at the same time, if $\xi_x$ is fixed to an arbitrary value (e.g., $1$). 
We usually keep $\alpha_x$, $\beta_x$ and $\xi_x$ in our formulas to allow easy comparison to traditional RZF formulations.
We assume the following normalization of the precoder:
\begin{align}
\Fm_x = \sqrt{K}\frac{\Mm_x}{\sqrt{\tr \lr\Mm_x^\H\Mm_x\rr}} \label{eq:iaRZF_TwoBSPrecoderNorm}
\end{align}
i.e., the sum energy of the precoder $\tr \lr\Fm_x^\H\Fm_x\rr$ is $K$.\footnote{It can be shown, using results from Appendix~\ref{ssec:iaRZF_appAsilomarPwrNorm} by taking $\chi_i=1$ $\forall i$, that this implies $\|\fv_{x,k}\|_2^2\to 1$, almost surely, under Assumption~\ref{as:iaRZF_cchain} for the given simplified system.
} 

\begin{remark}[Channel Scaling $1/N$]
	\label{rem:iaRZF_1overNfactor}
	The statistics of the channel matrices in this section incorporate the factor $1/N$, which simplifies comparisons with the later, more general, large-scale results (see Section~\ref{sec:iaRZF_Asilomar}).
	This can also be interpreted as transferring a scaling of the transmit power into the channel itself. 
	The precoder formulations presented in the current section can be simply rewritten to fit the more traditional statistics of $\hv_k \sim \CCC\NNN(0, \Id_N)$ and $\gv_k \sim \CCC\NNN(0, \varepsilon\Id_N)$, by using
	\begin{align*}
	\widetilde{\Mm}_x & = \lr \alpha_x\Hm_x\Hm_x^\H+\beta_x\hat{\Gm}_x\hat{\Gm}_x^\H + N\xi_x\Id  \rr^\mo \Hm_x 
	\end{align*}
	instead of $\Mm_x$.
	This equation shows that, under the chosen model, the regularization implicitly scales with $N$. However, one can either chose $\xi$ or $\alpha, \beta$ appropriately, to achieve any scaling.
\end{remark}

\subsection{Performance of Simple System}
\label{ssec:iaRZF_PerfSimpleSystem}
First, we compare the general performance of the proposed iaRZF scheme with classical approaches, i.e., single cell zero-forcing (ZF), maximum-ratio transmission (MRT) and RZF.
The rate of UT$_{x,k}$ can be defined as
\begin{align*}
r_{x,k} &= \log_2\lr 1 +  \frac{\text{Sig}_{x,k}}{\text{Int}_{x,k}^a + \text{Int}_{x,k}^r + 1} \rr
\end{align*}
where
$\text{Sig}_{x,k} 	= \rho_x \hv_{x,k}^\H \fv_{x,k} \fv_{x,k}^\H	\hv_{x,k}$, 
$\text{Int}_{x,k}^a	 = \rho_x \hv_{x,k}^\H \Fm_{x[k]} \Fm_{x[k]}^\H \hv_{x,k}$ and 
$\text{Int}_{x,k}^r	= \rho_{\bar{x}} \gv_{\bar{x},k}^\H \Fm_{\bar{x}} \Fm_{\bar{x}}^\H \gv_{\bar{x},k}$
denote the received signal power, received intra cell interference and received inter cell interference, respectively.

\begin{figure}
	\centering

   \begin{tikzpicture}[scale=0.8,font=\normalsize]
    \tikzstyle{every major grid}+=[style=densely dashed]
    \tikzstyle{every axis legend}+=[cells={anchor=west},fill=white,
        at={(0.02,0.98)}, anchor=north west, font=\normalsize ]
    \begin{axis}[
      xmin=-10,
      ymin=0.3,
      xmax=10,
      ymax=3,
      grid=major,
      scaled ticks=true,
   			xlabel={Per User Transmit Power to Noise ratio [dB]},   			
   			ylabel={Average Rate User$_{x,k}$ [bit/sec/Hz]},
      x post scale=1.1		
      ]
    \addplot[color=black!100, no marks, mark size=1.5pt,mark=x,line width=2pt] coordinates{
    (-15.000,0.162) (-14.000,0.198) (-13.000,0.240) (-12.000,0.291) (-11.000,0.350) (-10.000,0.418) (-9.000,0.496) (-8.000,0.585) (-7.000,0.687) (-6.000,0.796) (-5.000,0.922) (-4.000,1.060) (-3.000,1.213) (-2.000,1.382) (-1.000,1.563) (0.000,1.762) (1.000,1.976) (2.000,2.204) (3.000,2.443) (4.000,2.699) (5.000,2.964) (6.000,3.241) (7.000,3.528) (8.000,3.820) (9.000,4.122) (10.000,4.426) (11.000,4.740) (12.000,5.054) (13.000,5.372) (14.000,5.695) (15.000,6.017) 
    };
    \addlegendentry{ {iaRZF $\tau=0$} }
    \addplot[color=black!100, no marks, dashed, mark size=1.5pt,mark=o,line width=2pt] coordinates{
    (-15.000,0.162) (-14.000,0.197) (-13.000,0.239) (-12.000,0.289) (-11.000,0.347) (-10.000,0.413) (-9.000,0.490) (-8.000,0.576) (-7.000,0.671) (-6.000,0.774) (-5.000,0.891) (-4.000,1.016) (-3.000,1.151) (-2.000,1.297) (-1.000,1.448) (0.000,1.609) (1.000,1.773) (2.000,1.944) (3.000,2.107) (4.000,2.278) (5.000,2.434) (6.000,2.592) (7.000,2.733) (8.000,2.863) (9.000,2.984) (10.000,3.099) (11.000,3.187) (12.000,3.271) (13.000,3.335) (14.000,3.381) (15.000,3.451) 
    };
    \addlegendentry{ {iaRZF $\tau=0.5$} }
    \addplot[color=black!100, no marks, dotted, mark size=1.5pt,mark=o,line width=2pt] coordinates{
    (-15.000,0.161) (-14.000,0.196) (-13.000,0.236) (-12.000,0.284) (-11.000,0.339) (-10.000,0.401) (-9.000,0.472) (-8.000,0.547) (-7.000,0.630) (-6.000,0.715) (-5.000,0.808) (-4.000,0.903) (-3.000,0.999) (-2.000,1.096) (-1.000,1.192) (0.000,1.284) (1.000,1.370) (2.000,1.459) (3.000,1.525) (4.000,1.602) (5.000,1.655) (6.000,1.717) (7.000,1.758) (8.000,1.802) (9.000,1.838) (10.000,1.855) (11.000,1.877) (12.000,1.893) (13.000,1.903) (14.000,1.914) (15.000,1.933) 
    };
    \addlegendentry{ {iaRZF $\tau=1$} }    
    %
    \addplot[color=blue, no marks, mark size=2.5pt,mark=star,line width=2pt] coordinates{
    (-15.000,0.162) (-14.000,0.199) (-13.000,0.241) (-12.000,0.292) (-11.000,0.351) (-10.000,0.419) (-9.000,0.498) (-8.000,0.584) (-7.000,0.680) (-6.000,0.783) (-5.000,0.897) (-4.000,1.015) (-3.000,1.138) (-2.000,1.263) (-1.000,1.388) (0.000,1.511) (1.000,1.628) (2.000,1.744) (3.000,1.831) (4.000,1.934) (5.000,2.009) (6.000,2.084) (7.000,2.138) (8.000,2.197) (9.000,2.234) (10.000,2.271) (11.000,2.300) (12.000,2.319) (13.000,2.340) (14.000,2.354) (15.000,2.376) 
    };
    \addlegendentry{ {RZF} }	
    %
    \addplot[color=green, no marks, mark size=2.5pt,mark=o,line width=2pt] coordinates{
    (-15.000,0.128) (-14.000,0.159) (-13.000,0.196) (-12.000,0.240) (-11.000,0.294) (-10.000,0.357) (-9.000,0.430) (-8.000,0.515) (-7.000,0.609) (-6.000,0.715) (-5.000,0.831) (-4.000,0.953) (-3.000,1.080) (-2.000,1.210) (-1.000,1.343) (0.000,1.470) (1.000,1.594) (2.000,1.714) (3.000,1.807) (4.000,1.914) (5.000,1.992) (6.000,2.070) (7.000,2.127) (8.000,2.187) (9.000,2.227) (10.000,2.265) (11.000,2.295) (12.000,2.315) (13.000,2.337) (14.000,2.352) (15.000,2.375) 
    };
    \addlegendentry{ {ZF} }		
    %
    \addplot[color=red, no marks, mark size=2.5pt,mark=x,line width=2pt] coordinates{
    (-15.000,0.164) (-14.000,0.201) (-13.000,0.244) (-12.000,0.296) (-11.000,0.357) (-10.000,0.424) (-9.000,0.505) (-8.000,0.589) (-7.000,0.682) (-6.000,0.769) (-5.000,0.868) (-4.000,0.965) (-3.000,1.062) (-2.000,1.157) (-1.000,1.237) (0.000,1.324) (1.000,1.390) (2.000,1.460) (3.000,1.501) (4.000,1.558) (5.000,1.593) (6.000,1.625) (7.000,1.654) (8.000,1.677) (9.000,1.683) (10.000,1.702) (11.000,1.719) (12.000,1.723) (13.000,1.744) (14.000,1.744) (15.000,1.748) 
    };
    \addlegendentry{ {MRT} }	    
    \end{axis}
  \end{tikzpicture}   
	\caption[Simple system average user rate vs.~transmit power to noise ratio]{Average user rate vs.~transmit power to noise ratio ($N=160$, $K=40$, $\varepsilon=0.7$, $\rho_1=\rho_2=\rho$).}
	\label{fig:iaRZF_PerfSimpleiaRZF}
\end{figure}
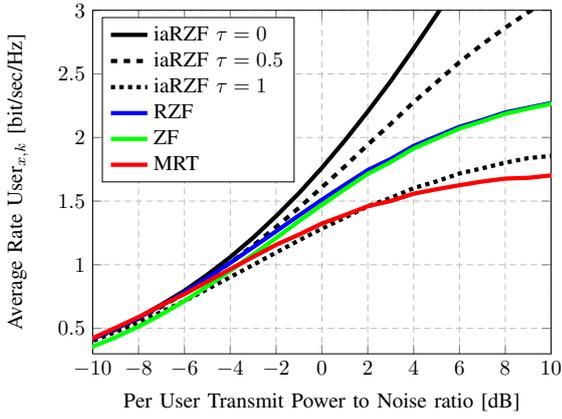
For comparison we used the following precoder formulations:
$\Mm_x^{\rm MRT} = \Hm_x$, $\Mm_x^{\rm ZF} = \Hm_x (\Hm_x^\H\Hm_x)^\mo$, $\Mm_x^{\rm RZF} = \Hm_x (\Hm_x^\H\Hm_x + \frac{K}{N\rho_x} \Id)^\mo$, where the regularization in $\Mm_x^{RZF}$ is chosen according to \cite{Bjoernson2014,Wagner2012a} and all precoders are normalized as in \eqref{eq:iaRZF_TwoBSPrecoderNorm}.
The iaRZF weights have been chosen to be $\alpha = \beta = N\cdot\rho_x$ and $\xi=1$, to simplify comparison with traditional RZF precoding.
The corresponding performance graphs, obtained by extensive Monte-Carlo (MC) simulations in the simplified system (for $N=160$, $K=40$, $\varepsilon=0.7$, $\rho_1=\rho_2=\rho$), are depicted in Figure~\ref{fig:iaRZF_PerfSimpleiaRZF}.

We observe that iaRZF generally outperforms the other schemes. This is not surprising, as the classical schemes do not take information about the interfered UTs into account. What is surprising, however, is the gain in performance even for very bad channel estimates (see curve for $\tau=0.5$). Only for extremely bad CSI we observe that iaRZF wastes energy due to a an improper choice of $\alpha,\,\beta$. Thus, it performs worse than the other schemes, that do not take $\tau$ into account for precoding. This problem can easily be circumvented by choosing proper weights that let $\beta\to0$ for $\tau\to1$; as will be elaborated later on.

\subsection[iaRZF for $\alpha_x,\beta_x\to \infty$]{Properties of iaRZF for $\boldsymbol{\alpha_x,\beta_x\to \infty}$}
\label{ssec:iaRZF_LimitBehavioriaRZF}
As has been briefly remarked by Hoydis \emph{et~al.} in \cite{Hoydis2013c}, the iaRZF weights $\alpha_x$ and $\beta_x$ should, intuitively, allow to project the transmitted signal in subspaces orthogonal to the UT$_x$'s (``own users'') and UT$_{\bar{x}}$'s (``other users'') channels, respectively.
This behavior, in the limit cases of $\alpha_x$ or $\beta_x \to \infty$, is treated analytically in this subsection.

\subsubsection{Finite Dimensional Analysis}
\label{sssec:iaRZF_LimitFinitedim}
Limiting ourselves to finite dimensional approaches and to perfect CSI ($\tau=0$), we can obtain the following insights.

First, we introduce the notation $\Pm_\Xm^\perp$ as a projection matrix on the space orthogonal to the column space of $\Xm$ and we remind ourselves that $\xi=1$ is still assumed. Following the path outlined in Appendix~\ref{ssec:iaRZF_ProofFinitDim} and assuming $\Hm_x^\H\Hm_x$ invertible (true with probability~$1$), one finds for $\alpha_x\to\infty$:
\begin{align*}
\alpha_x \Mm_x &  \alphaxto   \numberthis \label{eq:iaRZF_FDimiaRZF}
\Hm_x\lr\Hm_x^\H\Hm_x\rr^\mo - \\
& \Pm_{\Hm_x}^\perp\Gm_x\lr\beta_x^\mo\Id+\Gm_x^\H\Pm_{\Hm_x}^\perp\Gm_x\rr^\mo\Gm_x^\H\Hm_x\lr\Hm_x^\H\Hm_x\rr^\mo \,.
\end{align*}
Recall that the received signal at the UTs of BS$_x$ in our simple model, due to (only) the intra cell users, is given as $\yv_x^\text{intra} = \Hm_x^\H \Fm_x \sv_x$. Thus,
\begin{align*}
\yv_x^\text{intra} &\eqtextt{Lem~\ref{lem:iaRZF_Projection}}  \upsilon \Hm_x^\H \Hm_x\lr\Hm_x^\H\Hm_x\rr^\mo \sv_x = \upsilon\sv_x
\end{align*}
where the precoder normalization leaves a scaling factor $\upsilon$ that is independent of $\alpha_x$.  The Lemma~\ref{lem:iaRZF_Projection} used here can be found in Appendix~\ref{sec:iaRZF_LemmasTools}.
Hence, we see that for $\alpha_x\to\infty$ and $\beta_x$ bounded, the precoder acts similar to a traditional ZF precoder, i.e., the intra cell interference is completely suppressed in our system. The scaling factor includes the loss of energy caused by the alignment constraint, i.e., $\upsilon$ decreases the more ${\rm span}\{\Hm_x\}$ and ${\rm span}\{\Gm_x\}$ intersect. It remains to mention that due to the iaRZF definition, exact ZF can only be achieved in the limit for $N=K$, where $\Hm_x(\Hm_x^\H\Hm_x)^\mo = (\Hm_x\Hm_x^\H)^\mo\Hm_x$; assuming the inverses exist.

Looking now at the limit $\beta_x\to\infty$ and following Appendix~\ref{ssec:iaRZF_ProofFinitDim}, one arrives at
\begin{align}
\Mm_x  &\betaxto \nonumber \\
	&\ls\Pm_{\Gm_x}^\perp-\Pm_{\Gm_x}^\perp\Hm_x\lr\alpha_x^\mo\Id+\Hm_x^\H\Pm_{\Gm_x}^\perp\Hm_x\rr^\mo\Hm_x^\H\Pm_{\Gm_x}^\perp\rs\Hm_x  \nonumber \\ 
  =	& \check{\Hm}\lr \Id+\alpha\check{\Hm}^\H\check{\Hm}\rr^\mo \label{eq:iaRZF_FDimbetatoinfty_projectedchannel}
\end{align}
where we introduced $\check{\Hm}=\Pm_\Gm^\perp\Hm$, as the channel matrix $\Hm$ projected on the space orthogonal to the channels of $\Gm$.
One remembers that the received signal due to inter cell interference in our simple model is given as
\begin{align*}
\yv_x^\text{inter} = \Gm_{\bar{x}}^\H \Fm_{\bar{x}} \sv_{\bar{x}} 
\end{align*}
which, via \eqref{eq:iaRZF_FDimbetatoinfty_projectedchannel} in the case of $\Mm_{\overline{x}}$ and Lemma~\ref{lem:iaRZF_Projection}, directly gives $\yv_x^\text{inter}=0$. I.e., we see that for $\beta_{\overline{x}}\to\infty$ and $\alpha_{\overline{x}}$ bounded, the induced inter cell interference vanishes.

In \eqref{eq:iaRZF_FDimbetatoinfty_projectedchannel}, we finally see one of the main motivators for defining iaRZF, in the chosen form. Choosing $\beta_x=0$ gives the standard single cell RZF solution; choosing $\beta_x\to\infty$ gives an intuitively reasonable RZF precoder on projected channels that makes sure no interference is induced in the other cell.
It stands to reason that a sum-rate optimal solution can be found as a trade-off between these two extremes, by balancing induced interference and received signal power.

\subsubsection{Large-Scale Analysis}
\label{sssec:iaRZF_LimitRMT}
In order to find an appropriate expression of the sum-rate that does not rely on random quantities, we anticipate results from Subsection~\ref{ssec:iaRZF_AsymptoticEquivalent}. There we find a deterministic limit to which the random values of $\mathrm{SINR}_x$ almost surely converge, when $N,K \to \infty$; assuming $0 < c < \infty$. 
We can adapt the results from Theorem~\ref{theo:iaRZF_DEofSINR} to fit our the current simplified model (by choosing the parameters of the generalized Theorem~\ref{theo:iaRZF_DEofSINR} as $L=2, K_x=K, N_x=N, \chi^x_x=1, \chi^x_{\bar{x}}=\varepsilon, \tau^x_{\bar{x}}=\tau, \tau^x_x=0, \alpha^x_x=\alpha_{x}, \alpha^x_{\bar{x}}=\beta_{x}, \xi=1$, $P_x=\rho_x$, for $x\in\{1,2\}$). Doing so ultimately results in the following performance indicators $\mathrm{Sig}_{x} \xrightarrow[N,K\to+\infty]{\mathrm{a.s.}} \overline{\mathrm{Sig}}_{x}$ and $\mathrm{Int}_{x} \xrightarrow[N,K\to+\infty]{\mathrm{a.s.}} \overline{\mathrm{Int}}_{x} $, where
\begin{align*}
\overline{\mathrm{Sig}}_{x} &= P_x \lr 1 - \frac{c \alpha_x^2 e_x^2}{(1+\alpha_x e_x)^2} - \frac{c\beta_x^2 \varepsilon^2 e_x^2}{(1+\beta_x\varepsilon e_x)^2}  \rr  \\
\overline{\mathrm{Int}}_{x} &= 
\underbrace{P_x c \frac{1}{\lr 1+\alpha_x e_x\rr^2}}_{\text{from BS $x$}}
+ 
\underbrace{P_{\bar{x}} c \varepsilon \frac{1+2\beta_{\bar{x}}\varepsilon\tau^2e_{\bar{x}}+\beta_{\bar{x}}^2\varepsilon^2\tau^2e_{\bar{x}}^2 }{\lr 1+\beta_{\bar{x}} \varepsilon e_{\bar{x}}\rr^2}}_{\text{from BS $\bar{x}$}}	\label{eq:iaRZF_SimpleSystemInterf}	\numberthis \\
&\eqdef  \nonumber 	\overline{\mathrm{Int}}_{x}^{\text{BS}x} + \overline{\mathrm{Int}}_{x}^{\text{BS}\bar{x}}
\\	 					 
e_x &= \lr 1+ \frac{c\alpha_x}{1+\alpha_x e_x} + \frac{c\beta_x \varepsilon}{1+\beta_x\varepsilon e_x} \rr^\mo \label{eq:iaRZF_FPEsimpleSystem} \numberthis
\end{align*}
where $e_x$ is 
the unique non negative solution to the fixed point equation~\eqref{eq:iaRZF_FPEsimpleSystem}.
These expressions are precise in the large-scale regime ($N,K \to \infty$, $0<K/N<\infty$) and good approximations for finite dimensions.
As a consequence of the continuous mapping theorem (e.g., \cite{Couillet2011a}) the above finally implies $\mathrm{SINR}_x \xrightarrow[N,K\to+\infty]{\mathrm{a.s.}} \overline{\mathrm{SINR}}_x=\overline{\mathrm{Sig}}_{x} (\,\overline{\mathrm{Int}}_{x} + 1\,)^\mo$.

After realizing that $0<\liminf e_x<\limsup e_x<\infty$ for $K,N \to \infty$ (see Lemma~\ref{lem:iaRZF_ebound}), the large-scale formulations give the insights we already obtained from the finite dimensional analysis (see previous subsection). Slightly simplifying~\eqref{eq:iaRZF_SimpleSystemInterf} to reflect the perfect CSI case ($\tau=0$), one obtains 
\begin{align*}
&\lim_{\alpha_x\to\infty} \overline{\mathrm{Int}}_{x}^{\text{BS}x} = \lim_{\alpha_x\to\infty} P_x c \frac{1}{\lr 1+\alpha_x e_x\rr^2} = 0 \\
&\lim_{\beta_{\bar{x}}\to\infty} \overline{\mathrm{Int}}_{x}^{\text{BS}\bar{x}} = \lim_{\beta_{\bar{x}}\to\infty} P_{\bar{x}} c \frac{\varepsilon }{\lr 1+\beta_x \varepsilon e_x\rr^2} = 0
\end{align*}
i.e., for $\alpha_x\to\infty$ the intra cell interference vanishes and for $\beta_{\bar{x}}\to\infty$ the induced inter cell interference vanishes. Hence, at this point we have re-obtained the results from the previous subsection, where only finite dimensional techniques were used.

The large system formulation can now also be used to look at the important practical case of mis-estimation of the channels to the other cell's users. Employing again $0<\liminf e_x<\limsup e_x<\infty$ and~\eqref{eq:iaRZF_SimpleSystemInterf} leads to
\begin{align*}
&\lim_{\alpha_x\to\infty} P_x c \frac{1}{\lr 1+\alpha_x e_x\rr^2} = 0 \\
&\lim_{\beta_{\bar{x}}\to\infty} P_{\bar{x}} c \frac{(\beta_x^{-2}+2\varepsilon\tau^2e_x\beta_x^\mo+\varepsilon^2\tau^2e_x^2)\varepsilon }{\lr \beta_x^\mo+\varepsilon e_x\rr^2} = P_{\bar{x}} c \tau^2 \varepsilon
\end{align*}
i.e., for $\alpha_x\to\infty$ the intra cell interference still vanishes, under the assumption of perfect intra cell CSI. However, for $\beta_{\bar{x}}\to\infty$ the induced inter cell interference converges to $P_{\bar{x}} c \tau^2 \varepsilon$.
Unsurprisingly we see that, due to imperfect CSI, the induced inter cell interference cannot be completely canceled any more. 

\subsubsection{Large Scale Optimization}
\label{sssec:iaRZF_LimitOptimization}
One advantage of the large-scale approximation is the possibility to find asymptotically optimal weights for the limit behavior of iaRZF. 
However, to keep the calculations within reasonable effort, we limit the model to the case $P_1=P_2=P$  and $\alpha_1=\alpha_2=\alpha$, $\beta_1=\beta_2=\beta$.
Proceeding similar to the previous subsection, we obtain a formulation for the large-scale approximation of the (now equal) $\mathrm{SINR}$ values, when $\alpha\to\infty$. This is denoted 
$\overline{\mathrm{SINR}}^{\alpha\to\infty} = \overline{\mathrm{Sig}}^{\alpha\to\infty} \lr 1 + \overline{\mathrm{Int}}^{\alpha\to\infty}\rr^\mo$,
where
\begin{align*}
\overline{\mathrm{Sig}}^{\alpha\to\infty} &= 
P \lr 1 - c - \frac{c\beta^2 \varepsilon^2 e^2}{(1+\beta\varepsilon e)^2}  \rr  \\
\overline{\mathrm{Int}}^{\alpha\to\infty} &= 	 	
P c \varepsilon \frac{1+2\beta\varepsilon\tau^2e+\beta^2\varepsilon^2\tau^2e^2}{\lr 1+\beta \varepsilon e\rr^2}
\end{align*} 
and
\begin{align*}
e \eqdef e^{\alpha\to\infty} &= 
\lr 1+ \frac{c}{e} + \frac{c\beta \varepsilon}{1+\beta\varepsilon e} \rr^\mo \numberthis \,. \label{eq:iaRZF_SimpleModelLimitFPEe}
\end{align*}

The optimal values of the weight $\beta$ in limit case $\alpha\to\infty$ can be found by solving $ \partial \overline{\mathrm{SINR}}^{\alpha\to\infty} / \partial \beta = 0$. 
This leads to (see Appendix~\ref{ssec:iaRZF_appProofSimpleOptimizationalphainfty})
\begin{align}
\beta^{\alpha\to\infty}_{opt} = \frac{P (1-\tau^2) }{P c\varepsilon\tau^2+1} \,.
\label{eq:iaRZF_SimpSysbetaoptinf}
\end{align}
In other words, in the perfect CSI case ($\tau=0$), one should choose $\beta$ equal to the transmit power of the BSs. It also shows how one should scale $\beta$ in between the two obvious solutions, i.e., full weight on the interfering channel information for perfect CSI and no weight under random CSI ($\tau=1$).  
We remark that the interference channel gain factor $\varepsilon$ is also implicitly included in the precoder. Thus for $\varepsilon \to 0$, we have $\beta(\trace\hat{\Gm}_x^\H\hat{\Gm}_x)^2 \to 0$, 
while $\beta$ remains bounded. Hence no energy is wasted to precode for non-existent interference, as one would expect.

The same large-scale optimization can also be carried out for the limit $\beta\to\infty$. The SINR optimal weight for $\alpha$ can be found as 
(similar to Appendix~\ref{ssec:iaRZF_appProofSimpleOptimizationalphainfty})
\begin{align}
\alpha^{\beta\to\infty}_{opt} = \frac{P}{P c\varepsilon\tau^2+1} = \frac{1}{c\varepsilon\tau^2+1/P} \,.
\label{eq:iaRZF_SimpSysalphaoptinf}
\end{align}
Thus, like in the perfect CSI case ($\tau=0$), one should choose $\alpha$ equal to the transmit power of the BSs.
The implications for other limit-cases are not so clear.
We see that increasing the transmit power also increases the weight $\alpha$, up to the maximum value of $ 1/(c\varepsilon\tau^2)$.
The weight reduces as the interference worsens, i.e., when $\tau^2$, $\varepsilon$ grow. This makes sense, as the precoder would give more importance on the interfering channel (by indirectly increasing $\beta$ via normalization). The weight is also reduced, if the cell performance is expected to be bad, i.e., $c$ approaches $1$, which makes sense from a sum-rate optimization point of view.

Finally, we can easily calculate the SINR in the limit of both $\alpha$ and $\beta$ independently tending to infinity, as
\begin{align*}
\overline{\mathrm{SINR}}^{\alpha,\beta\to\infty} = \frac{P \lr 1-2c\rr}{{P c\varepsilon\tau^2+1}} \,.
\end{align*}


The rationale behind all analyses in this section is that optimal weights in the limiting case often make for good heuristic approximations in more general cases. 
For instance, one can
re-introduce the weights, found under the large-scale assumption, into the finite dimensional limit formulations.
Particularly interesting for this approach is combining \eqref{eq:iaRZF_SimpSysbetaoptinf} with \eqref{eq:iaRZF_FDimiaRZF} to find an ``heuristic interference aware zeroforcing'' precoder:
\begin{align*}
\Mm_x^{iaZF} & =
\Hm_x\lr\Hm_x^\H\Hm_x\rr^\mo - 
\Pm_{\Hm_x}^\perp\hat{\Gm}_x \\
& \times \lr\frac{P c\varepsilon\tau^2+1}{P (1-\tau^2) }\Id+\hat{\Gm}_x^\H\Pm_{\Hm_x}^\perp\hat{\Gm}_x\rr^\mo\hat{\Gm}_x^\H\Hm_x\lr\Hm_x^\H\Hm_x\rr^\mo \,.
\end{align*}
To finish this section, we remark that doing these optimizations without the assumption of arbitrarily well estimated intra cell channels, is still an open problem.

\subsubsection{Graphical Interpretation of the Results}
\label{sssec:iaRZF_LimitFiniteRMTFigures}
Here we show the implications of the previous subsection on the system performance of our simple model.
Of particular interest to us are comparisons of the heuristic scheme with, numerically obtained, sum-rate optimal weights.

\begin{figure}
	\centering

   \begin{tikzpicture}[scale=0.8,font=\normalsize]
    \tikzstyle{every major grid}+=[style=densely dashed]
    \tikzstyle{every axis legend}+=[cells={anchor=west},fill=white,
        at={(0.05,0.02)}, anchor=south west, font=\normalsize ]
    \begin{axis}[
      xmin=0,
      ymin=1.6,
      xmax=1,
      ymax=2.8,
      grid=major,
      scaled ticks=true,
   			xlabel={CSI randomness $\tau$},
   			ylabel={Average Rate User$_{x,k}$ [bit/sec/Hz]},
      x post scale=1.1		
      ]
	\addplot[color=black!100, mark size=4pt,mark options={solid, line width=1pt}, mark=x,line width=2pt] coordinates{
    (0.000,2.777) (0.050,2.773) (0.100,2.759) (0.150,2.737) (0.200,2.706) (0.250,2.669) (0.300,2.626) (0.350,2.578) (0.400,2.526) (0.450,2.472) (0.500,2.416) (0.550,2.359) (0.600,2.303) (0.650,2.248) (0.700,2.194) (0.750,2.144) (0.800,2.098) (0.850,2.056) (0.900,2.021) (0.950,1.995) (1.000,1.984) 
    };
    \addlegendentry{ {$(\alpha^{ls}_{opt}, \beta^{ls}_{opt})$} }    
    %
    \addplot[color=blue, dash pattern=on 3pt off 1pt on 1pt off 1pt, mark options={solid, line width=1pt}, mark size=4pt, mark=o, line width=2pt] coordinates{
    (0.000,2.777) (0.050,2.773) (0.100,2.759) (0.150,2.737) (0.200,2.706) (0.250,2.669) (0.300,2.626) (0.350,2.578) (0.400,2.526) (0.450,2.472) (0.500,2.416) (0.550,2.359) (0.600,2.303) (0.650,2.248) (0.700,2.194) (0.750,2.144) (0.800,2.098) (0.850,2.056) (0.900,2.021) (0.950,1.995) (1.000,1.984) 
    };
    \addlegendentry{ {$(\alpha^{\beta\to\infty}_{opt}, \beta^{\alpha\to\infty}_{opt})$} }		
    %
    \addplot[color=green, dashed, no marks, mark size=2.5pt,mark=o,line width=2pt] coordinates{
    (0.000,2.701) (0.050,2.696) (0.100,2.682) (0.150,2.659) (0.200,2.628) (0.250,2.590) (0.300,2.545) (0.350,2.496) (0.400,2.443) (0.450,2.387) (0.500,2.330) (0.550,2.272) (0.600,2.215) (0.650,2.159) (0.700,2.106) (0.750,2.055) (0.800,2.009) (0.850,1.968) (0.900,1.933) (0.950,1.908) (1.000,1.898) 
    };
    \addlegendentry{ {$(\infty, \beta^{\alpha\to\infty}_{opt})$} }
    black!40
    \addplot[color=red, densely dashed, no marks, mark size=2.5pt,mark=x,line width=2pt] coordinates{
    (0.000,2.676) (0.050,2.671) (0.100,2.656) (0.150,2.632) (0.200,2.599) (0.250,2.558) (0.300,2.511) (0.350,2.457) (0.400,2.400) (0.450,2.338) (0.500,2.274) (0.550,2.209) (0.600,2.142) (0.650,2.076) (0.700,2.009) (0.750,1.944) (0.800,1.879) (0.850,1.816) (0.900,1.755) (0.950,1.695) (1.000,1.637) 
    };
    \addlegendentry{ {$(\alpha^{\beta\to\infty}_{opt}, \infty)$} }		
    black!20
    \addplot[color=cyan, dotted, no marks, mark size=2.5pt,mark=o,line width=2pt] coordinates{
    (0.000,1.984) (0.050,1.984) (0.100,1.984) (0.150,1.984) (0.200,1.984) (0.250,1.984) (0.300,1.984) (0.350,1.984) (0.400,1.984) (0.450,1.984) (0.500,1.984) (0.550,1.984) (0.600,1.984) (0.650,1.984) (0.700,1.984) (0.750,1.984) (0.800,1.984) (0.850,1.984) (0.900,1.984) (0.950,1.984) (1.000,1.984) 
    };
    \addlegendentry{ {$(\alpha^{ls}_{opt0}, 0)$} }
    \end{axis}
  \end{tikzpicture}
	\caption[Simple system average user rate vs.~CSI quality for adaptive precoder weights.]{Average user rate vs.~CSI quality for adaptive precoder weights ($N=160$, $K=40$, $\varepsilon=0.7$, $P=10$dB).}
	\label{fig:iaRZF_Figure3b}
\end{figure}
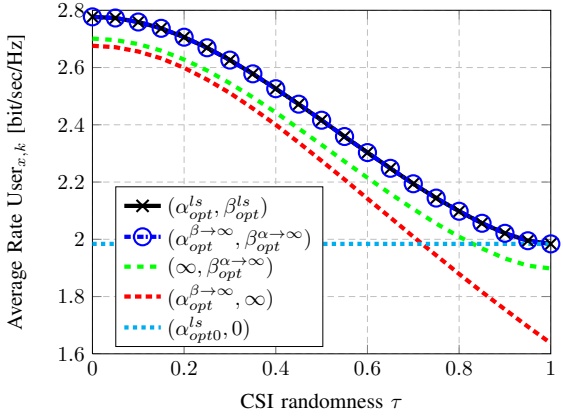
In Figure~\ref{fig:iaRZF_Figure3b}, we analyze the average UT rate with respect to CSI randomness ($\tau$), for different sets of precoder weights $(\alpha,\beta)$, that (mostly) adapt to the available CSI quality.
The values $(\alpha^{ls}_{opt}$ and $\beta^{ls}_{opt})$ are obtained using $4$D grid search, thereby renouncing the $\alpha_1=\alpha_2=\alpha$ and $\beta_1=\beta_2=\beta$ restrictions from before. 
Crucially, we see that the performance under $(\alpha^{ls}_{opt}, \beta^{ls}_{opt})$ and $(\alpha^{\beta\to\infty}_{opt}, \beta^{\alpha\to\infty}_{opt})$ is virtually the same. 
\\
The plot also contains the pair $(\alpha^{ls}_{opt0}, 0)$, which corresponds to single cell RZF precoding. The weight $\alpha^{ls}_{opt0}$ is again found by grid search. 
The performance is constant, as the precoder does not take the interfering channel into account.
However, we see that the optimally weighted iaRZF is equal to single cell RZF, when the channel estimation is purely random.

\begin{figure}
	\centering

   \begin{tikzpicture}[scale=0.8,font=\normalsize]
    \tikzstyle{every major grid}+=[style=densely dashed]
    \tikzstyle{every axis legend}+=[cells={anchor=west},fill=white,
        at={(0.02,0.02)}, anchor=south west, font=\normalsize ]
    \begin{axis}[
      xmin=0,
      ymin=1.6,
      xmax=1,
      ymax=2.9,
      grid=major,
      scaled ticks=true,
   			xlabel={CSI randomness $\tau$},
   			ylabel={Average Rate User$_{x,k}$ [bit/sec/Hz]},
      x post scale=1.1			
      ]
	\addplot[color=black!100, no marks, mark size=1.5pt,mark=x,line width=2pt] coordinates{
    (0.000,2.777) (0.050,2.773) (0.100,2.759) (0.150,2.737) (0.200,2.706) (0.250,2.669) (0.300,2.626) (0.350,2.578) (0.400,2.526) (0.450,2.472) (0.500,2.416) (0.550,2.359) (0.600,2.303) (0.650,2.248) (0.700,2.194) (0.750,2.144) (0.800,2.098) (0.850,2.056) (0.900,2.021) (0.950,1.995) (1.000,1.984) 
    };
    \addlegendentry{ {$\beta = \beta^{\alpha\to\infty}_{opt}$} }    
    %
    \addplot[color=blue, dash pattern=on 3pt off 1pt on 1pt off 1pt, no marks, mark size=2.5pt,mark=o,line width=2pt] coordinates{
    (0.000,2.399) (0.050,2.398) (0.100,2.394) (0.150,2.387) (0.200,2.378) (0.250,2.365) (0.300,2.351) (0.350,2.334) (0.400,2.314) (0.450,2.292) (0.500,2.268) (0.550,2.242) (0.600,2.215) (0.650,2.185) (0.700,2.155) (0.750,2.122) (0.800,2.089) (0.850,2.055) (0.900,2.020) (0.950,1.984) (1.000,1.948) 
    };
    \addlegendentry{ {$\beta = 1$} }		
    %
    \addplot[color=green, dashed, no marks, mark size=2.5pt,mark=x,line width=2pt] coordinates{
    (0.000,2.744) (0.050,2.741) (0.100,2.729) (0.150,2.710) (0.200,2.684) (0.250,2.652) (0.300,2.613) (0.350,2.570) (0.400,2.522) (0.450,2.470) (0.500,2.415) (0.550,2.359) (0.600,2.300) (0.650,2.240) (0.700,2.180) (0.750,2.119) (0.800,2.059) (0.850,1.999) (0.900,1.940) (0.950,1.882) (1.000,1.825) 
    };
    \addlegendentry{ {$\beta = 5$} }		
    %
    \addplot[color=red, densely dashed, no marks, mark size=2.5pt,mark=o,line width=2pt] coordinates{
    (0.000,2.777) (0.050,2.773) (0.100,2.759) (0.150,2.736) (0.200,2.706) (0.250,2.668) (0.300,2.623) (0.350,2.573) (0.400,2.518) (0.450,2.459) (0.500,2.398) (0.550,2.334) (0.600,2.269) (0.650,2.204) (0.700,2.138) (0.750,2.073) (0.800,2.008) (0.850,1.945) (0.900,1.883) (0.950,1.822) (1.000,1.763) 
    };
    \addlegendentry{ {$\beta = 10$} }
    \addplot[color=cyan, dotted, no marks, mark size=2.5pt,mark=o,line width=2pt] coordinates{
    (0.000,2.701) (0.050,2.696) (0.100,2.681) (0.150,2.657) (0.200,2.624) (0.250,2.583) (0.300,2.535) (0.350,2.482) (0.400,2.423) (0.450,2.362) (0.500,2.297) (0.550,2.231) (0.600,2.165) (0.650,2.097) (0.700,2.031) (0.750,1.964) (0.800,1.899) (0.850,1.836) (0.900,1.774) (0.950,1.714) (1.000,1.655) 
    };
    \addlegendentry{ {$\beta = 100$} }
    \end{axis}
    \begin{axis}[
      xmin=0,
	  ymin=0,
      xmax=1,
      ymax=15,
      ylabel near ticks, 
      yticklabel pos=right,
      axis x line=none,
      ylabel={$\beta_{opt}^{ls}$},
      ylabel shift = -15 pt,
      scaled ticks=true,
      clip=false	,
      x post scale=1.1,
      legend style={at={(0.98,0.98)}, anchor=north east, font=\normalsize, cells={anchor=west},fill=white}
      ]%
        \addplot[color=orange, no marks ,line width=1pt] coordinates{
        (0.000,10.001) (0.050,9.932) (0.100,9.730) (0.150,9.405) (0.200,8.972) (0.250,8.451) (0.300,7.862) (0.350,7.226) (0.400,6.563) (0.450,5.889) (0.500,5.218) (0.550,4.561) (0.600,3.927) (0.650,3.320) (0.700,2.746) (0.750,2.205) (0.800,1.698) (0.850,1.226) (0.900,0.786) (0.950,0.378) (1.000,0.000) 
        } 
		node at (axis cs:0.7,1.8) { {\BLACK $\beta_{opt}^{ls}$} };
		\addlegendentry{ {$\beta_{opt}^{ls}$} };
		%
		%
		\draw[black!30,dashed,line width=1pt] (axis cs:0.7,2.3) circle[radius=1.5em];
		\draw[black!30,dashed,line width=1pt] (axis cs:1.05,7.5) circle[radius=1.5em];		
		\draw[black!30,dashed,line width=1pt] [shorten >=1.5em, shorten <=1.5em] (axis cs:0.7,2.3) -- (axis cs:1.05,7.5);
    \end{axis}  
  \end{tikzpicture} 
	\caption[Simple system average user rate vs.~CSI quality for constant precoder weights.]{Average user rate vs.~CSI quality for constant precoder weights ($N=160$, $K=40$, $\varepsilon=0.7$, $\alpha = \alpha^{\beta\to\infty}_{opt}$, $P=10$dB).}
	\label{fig:iaRZF_Figure3a}
\end{figure}
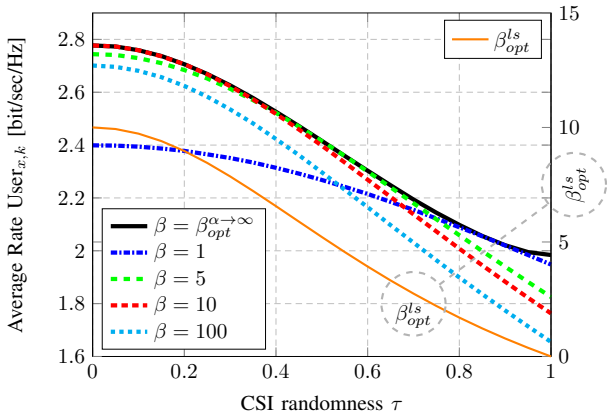
In Figure~\ref{fig:iaRZF_Figure3a}, we illustrate the effect of (sub-optimally, but conveniently) choosing a constant value for $\beta$. 
We set $\alpha = \alpha^{\beta\to\infty}_{opt}$ for all curves and also give the familiar $(\alpha^{\beta\to\infty}_{opt}, \beta^{\alpha\to\infty}_{opt})$ curve as a benchmark. Furthermore, the actual value of $\beta^{ls}_{opt}$ is given on a second axis to illustrate how one would need to adapt $\beta$ for optimal average rate performance.
Overall one observes that a constant value for $\beta$ is (unsurprisingly) only acceptable for a limited region of the CSI quality spectrum.
Small values of $\beta$ fit well for large $\tau$, middle values fit well for small $\tau$. Overly large (or small) $\beta$s do not reach optimal performance in any region.

The encouraging performance of iaRZF using the optimal weights derived under limit assumptions, paired with the promise of simple and intuitive insights, provides motivation for the next section. There we  apply the iaRZF scheme to a more general system.

\section{General System for iaRZF Analysis}
\label{sec:iaRZF_Asilomar}

\subsection{System Model}
\label{ssec:iaRZF_Systemmodel}
In the following, we analyze cellular downlink multi-user MIMO systems, of the more general type illustrated in Fig.~\ref{fig:iaRZF_DownlinkLCells}. Each of the $L$ cells consists of one BS associated with a number of single antenna UTs.
\begin{figure}
	\centering
	\includegraphics[width=0.4\textwidth]{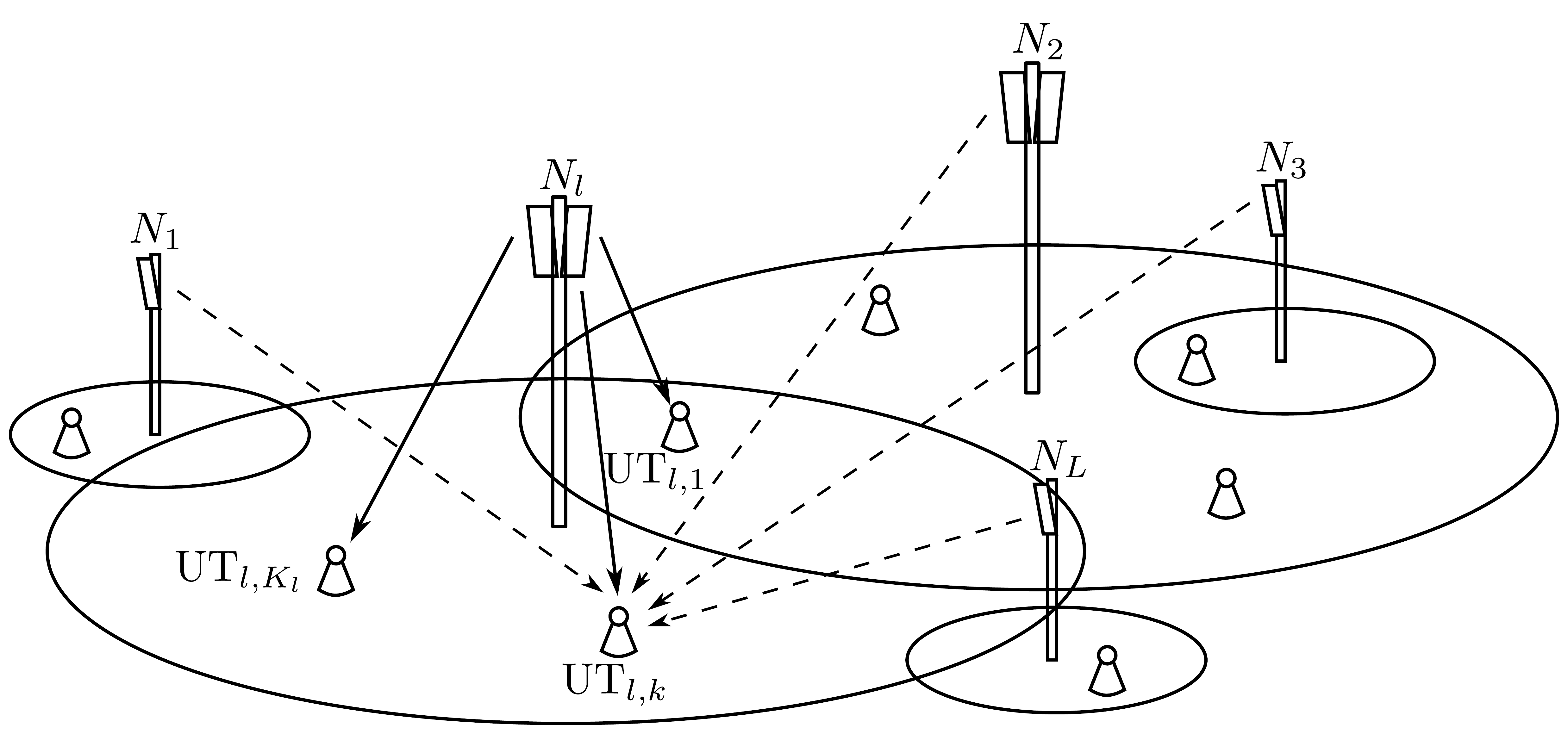}
	\vspace{0.51em}
	\caption{Illustration of a general heterogeneous downlink system.}
	\label{fig:iaRZF_DownlinkLCells}
\end{figure}
In more detail, the $l$th BS is equipped with $N_l$ transmit antennas and serves $K_l$ UTs. We generally set $N_l \geq K_l$ in order to avoid scheduling complications. We assume transmission on a single narrow-band carrier, full transmit-buffers, and universal frequency reuse among the cells.

The $l$th BS transmits a data symbol vector $\sv_l=[s_{l,1}, \ldots, s_{l,K_l}]^\trans$ intended for its $K_l$ uniquely associated UTs. This BS uses the linear precoding matrix $\Fm^l_l \in \mathbb{C}^{N_l \times K_l}$, where the columns $\fv^l_{l,k}\in \mathbb{C}^{N_l}$ constitute the precoding vectors for each UT. We note that BSs do not directly interact with each other and users from other cells are explicitly not served. Thus, the received signal $y_{l,k} \in \CC$ at the $k$th UT in cell $l$ is
\begin{align*}
y_{l,k} &=  \sqrt{\chi^l_{l,k}} (\hv^l_{l,k})^\H \fv^l_{l,k} s_{l,k}
+ \sum_{k^\prime \neq k} \sqrt{\chi^l_{l,k}} (\hv^l_{l,k})^\H \fv^l_{l,k^\prime} s_{l,k^\prime}  \\
& \qquad + \sum_{m\neq l} \sqrt{\chi^m_{l,k}} (\hv^m_{l,k})^\H \Fm^m_m \sv_m + n_{l,k} 
\end{align*}
where $n_{l,k} \sim \CCC\NNN(0,1)$ is an additive noise term. The transmission symbols are chosen from a Gaussian codebook, i.e., $s_{l,k} \sim \CCC\NNN(0,1)$. 
We assume block-wise small scale Rayleigh fading, thus the channel vectors are modeled as $\hv^m_{l,k} \sim \CCC\NNN(\zerov,\frac{1}{N_m}\Id_{N_m})$. The path-loss and other large-scale fading effects are incorporated in the $\chi^m_{l,k}$ factors.
The scaling factor $\frac{1}{N_m}$ in the fading variances is of technical nature and utilized in the asymptotic analysis. It can be canceled for a given arbitrarily sized system by modifying the transmission power accordingly; similar to Remark~\ref{rem:iaRZF_1overNfactor}.
Our setting here assumes that interference from other cells dominates w.r.t.\ other types of rate limitations, such as pilot contamination, which is thus not accounted for in our system model. According to recent works, this assumption is sensible if one considers practical ranges of antennas (on the order of $100$s) \cite{Bjoernson2014MassiveMaximal}, in systems with optimized pilot-reuse \cite{Yang2013a}, and also when using non-ideal hardware \cite[Figure~14]{bjornson2013massive}.

\subsection{Imperfect Channel State Information}

The UTs are assumed to perfectly estimate the respective channels to their serving BS, which enables coherent reception. This is reasonable, even for moderately fast traveling users, if proper downlink reference signals are alternated with data symbols.
Generally, downlink CSI can be obtained using either a time-division duplex protocol where the BS acquires channel knowledge from uplink pilot signaling \cite{Hoydis2013a} and using channel reciprocity or a frequency-division duplex protocol, where temporal correlation is exploited as in \cite{Choi2014a}. 
In both cases, the transmitter usually has imperfect knowledge of the instantaneous channel realizations, e.g., due to imperfect pilot-based channel estimation, delays in the acquisition protocols, or user mobility. 
To model imperfect CSI without making explicit assumptions on the acquisition protocol, we employ the generic Gauss-Markov formulation (see, e.g.,
\cite{Wagner2012a,Wang2006a,Nosrat2011a}) and we define the estimated channel vectors $\hat{\hv}^{m}_{l,k} \in \CC^{N_m}$ to be
\begin{align}
\hat{\hv}^{m}_{l,k} = \sqrt{\chi^m_{l,k}}\ls \sqrt{(1-(\tau^m_l)^2)}\hv^{m}_{l,k} + \tau^m_l\tilde{\hv}^{m}_{l,k} \rs
\label{eq:iaRZF_estimatedchannel}
\end{align}
where $\tilde{\hv}^{m}_{l,k}\sim\CCC\NNN(0,\frac{1}{N_m}\Id_{N_m})$ is the normalized independent estimation error.
Using this formulation, we can set the accuracy of the channel acquisition between the UTs of cell $l$ and the BS of cell $m$ by selecting $\tau^m_l \in [0,1]$; a small value for $\tau^m_l$ implies a good estimate. Furthermore, we remark that these choices imply $\hat{\hv}^{m}_{l,k} \sim \CCC\NNN(0,\chi^m_{l,k}\frac{1}{N_m}\Id_{N_m})$. For convenience later on, we define the aggregated estimated channel matrices as $\hat{\Hm}^m_l = [\hat{\hv}^{m}_{l,1}, \ldots, \hat{\hv}^{m}_{l,K_l}] \in \CC^{N_m \times K_l} $.

\subsection{iaRZF and Power Constraints}

Following the promising results observed in Section~\ref{sec:iaRZF_SimpleiaRZF}, we continue our analysis of the iaRZF precoding matrices $\Fm^m_m, \quad m=1, \ldots, L$, introduced in~\eqref{eq:iaRZF_iaRZFprecoder}. For some derivations, it will turn out to be useful to restate this precoder as
\begin{align*}
\Fm^m_m & = \lr \alpha^m_m \hat{\Hm}^m_m (\hat{\Hm}^m_m)^\H + \Zm^m + \xi_m \Id_{N_m} \rr^\mo \hat{\Hm}^m_m \nu_m^\OH
\end{align*}
where $\Zm^m = \sum_{l\neq m}\alpha^m_l \hat{\Hm}^m_l (\hat{\Hm}^m_l)^\H$. The $\alpha^m_l$ can be considered as weights pertaining to the importance one wishes to attribute to the respective estimated channel. We remark that the regularization parameter $\xi_m$ is usually chosen to be the number of users over the total transmit power \cite{Bjoernson2014} in classical RZF.
The factors $\nu_m$ are used to fulfill the average per UT transmit power constraint $P_m$\footnote{We remark that choosing $P_m$ of order $1$ will assure proper scaling of all terms of the SINR in the following (see \eqref{eq:iaRZF_FDimSINR}).}, pertaining to BS $m$:
\begin{align}
\frac{1}{K_m} \trace\ls \Fm^m_m (\Fm^m_m)^\H \rs = P_m \,. \label{eq:iaRZF_FDimPwrCtrl}
\end{align}

\subsection{Performance Measure}
\label{ssec:iaRZF_SINRfiniteDim}
Most performance measures in cellular systems are functions of the SINRs at each UT; e.g., (weighted) sum-rate and outage probability. 
Under the treated system model, the expected (w.r.t.\ the transmitted symbols $s^{(l)}_{l,k}$) received signal power at the $k$th UT of cell $l$, i.e., UT$_{l,k}$, is 
\begin{align*}
\mathrm{Sig}^{(l)}_{l,k} 	
&= \chi^l_{l,k} (\hv^l_{l,k})^\H \fv^l_{l,k} (\fv^l_{l,k})^\H \hv^l_{l,k} \,. \numberthis 	\label{eq:iaRZF_FDimSig}
\end{align*}
Similarly, the interference power is
\begin{align*}
\mathrm{Int}^{(l)}_{l,k} 	
&= \sum_{m\neq l} \chi^m_{l,k} (\hv^m_{l,k})^\H \Fm^m_m (\Fm^m_m)^\H \hv^m_{l,k}  \\
&\qquad + \chi^l_{l,k} (\hv^l_{l,k})^\H \Fm^l_{l[k]} (\Fm^l_{l[k]})^\H \hv^l_{l,k}  \numberthis
\label{eq:iaRZF_FDimInt} 
\end{align*}
where
\begin{align}
\Fm^l_{l[k]} = \lr \alpha^l_l \hat{\Hm}^l_l (\hat{\Hm}^l_l)^\H + \Zm^l + \xi_l \Id_{N_l} \rr^\mo \hat{\Hm}^l_{l[k]} \nu_l^\OH
\end{align}
and $\hat{\Hm}^l_{l[k]}$ is $\hat{\Hm}^l_{l}$ with its $k$th column removed.
Hence, the SINR at UT$_{l,k}$ can be expressed as
\begin{align}
\mathrm{SINR}_{l,k} = \mathrm{Sig}^{(l)}_{l,k} \lr \mathrm{Int}_{l,k} + 1 \rr^\mo \,. \label{eq:iaRZF_FDimSINR}
\end{align}

In the following, we focus on the sum-rate, which is a commonly used performance measure utilizing the SINR values and straightforward to interpret. Under the assumption that interference is treated as noise, the ergodic sum-rate is expressed as
\begin{align*}
R_{sum} = \expect \sum_{l,k} r_{l,k} = \expect \sum_{l,k} \log (1 +\mathrm{SINR}_{l,k}) 
\end{align*}
where SINRs are random quantities defined by the system model.

\subsection{Deterministic Equivalent of the SINR}
\label{ssec:iaRZF_AsymptoticEquivalent}
In order to obtain tractable and insightful expressions of the system performance, we propose a large scale approximation. This allows us to state the sum-rate expression in a deterministic and compact form that can readily be interpreted and optimized.
Also, the large system approximations are accurate in both massive MIMO systems and conventional small-scale MIMO of tractable size, as will be evidenced later via simulations (see Subsection~\ref{ssec:iaRZF_GenModNumerics}).
In certain special cases, optimizations of such approximations w.r.t.\ many performance measures, can be carried out analytically (see for example \cite{Wagner2012a}). In almost all cases, optimizations can be done numerically. 
We will derive a deterministic equivalent (DE) of the SINR values that allows for a large scale approximation of the sum-rate expression in~\eqref{eq:iaRZF_FDimSINR}. DEs are preferable to standard limit calculations, as they are precise in the limit case, they are also defined for finite dimensions and they provably approach the random quantity for increasing dimensions.
Introducing the ratio $c_i = K_i/N_i$, we make the following technical assumption in order to obtain a DE. 
\begin{assumption}
	\label{as:iaRZF_cchain}
	$N_i,K_i \to \infty$, such that for all $i$ we have
	\begin{equation*}
	0 < \liminf c_i \leq \limsup c_i < \infty \,.
	\end{equation*}
	This asymptotic regime is denoted $N\to\infty$ for brevity.
\end{assumption}
In other words, we require for $N_i$ and $K_i$ to grow large at the same speed.
By extending the analytical approach in \cite{Wagner2012a} and \cite{Hoydis2013a} to the SINR expression in \eqref{eq:iaRZF_FDimSINR}, we obtain a DE of the SINR, which is denoted $\overline{\mathrm{SINR}}_{l,k}$ in the following.

\begin{theorem}[Deterministic Equivalent of the SINR]
	\label{theo:iaRZF_DEofSINR}
	Under {\bf A-\ref{as:iaRZF_cchain}}, we have
	\begin{align*}
	\mathrm{SINR}_{l,k} - \overline{\mathrm{SINR}}_{l,k} \xrightarrow[N \to \infty]{\mathrm{a.s.}} 0 \,.
	\end{align*}
	Here
	\begin{align*}
	\overline{\mathrm{SINR}}_{l,k} =  \overline{\mathrm{Sig}}^{(l)}_{l,k} \lr \overline{\mathrm{Int}}_{l,k} + 1 \rr^\mo
	\end{align*}
	with
	\begin{align*}
	\overline{\mathrm{Sig}}^{(l)}_{l,k} &=  \overline{\nu}_l (\chi^l_{l,k})^2 e_{l}^2 \lr 1-(\tau^l_l)^2 \rr (y^l_{l,k})^2 
	\end{align*}
	and
	\begin{align*}
	& \overline{\mathrm{Int}}_{l,k}  =  \sum^L_{m=1} \overline{\nu}_m
	\lr 1+2 x^m_{l,k} e_{m}+\alpha^m_l\chi^m_{l,k} x^m_{l,k} e_{m}^2 \rr \chi^m_{l,k} g_{m}  (y^m_{l,k})^2
	\end{align*}
	given $x^m_{l,k} = \alpha^m_l\chi^m_{l,k}(\tau^m_l)^2$.
	The parameter $\overline{\nu}_m$, the abbreviations $g_{m}$ and $y^m_{l,k}$, as well as the corresponding fixed-point equation $e_{m}$ and $e_{m}^\prime$ are given in the following. 
	
	First, we define $e_{m}$ to be the unique positive solution of the fixed-point equation
	\begin{align}
	& e_{m} = \label{eq:iaRZF_e} \\
	& \lr \xi_m
	+ \frac{1}{N_m}\sum_{j=1}^{K_m} \alpha^m_m\chi^m_{m,j} y^m_{m,j}
	+ \frac{1}{N_m}\sum_{l\neq m}\sum_{k=1}^{K_l} \alpha^m_l\chi^m_{l,k}y^m_{l,k}
	\rr^\mo \nonumber
	\end{align}
	where $y^m_{l,k} = \lr 1+\alpha^m_l \chi^m_{l,k} e_{m}  \rr^\mo$. 
	We also have
	$
	\overline{\nu}_m = {P_m K_m}/ \lr N_m g_{m} \rr
	$
	with
	\begin{align*}
	g_{m} &= - \frac{1}{N_m}\sum_{j=1}^{K_m} \chi^m_{m,j}e_{m}^\prime (y^m_{m,k})^2 
	\end{align*}
	and $e_{m}^\prime$ can be found directly, once $e_{m}$ is known:
	\begin{align}
	e_{m}^\prime  &=
	\Bigg[
	\frac{1}{N_m}\sum_{j=1}^{K_m} (\alpha^m_m)^2(\chi^m_{m,j})^2 (y^m_{m,j})^2 \nonumber \\
	&\qquad +
	\frac{1}{N_m}\sum_{l\neq m}\sum_{k=1}^{K_l} (\alpha^m_l)^2(\chi^m_{l,k})^2 (y^m_{l,k})^2
	-
	e_{m}^\mt
	\Bigg]^\mo \,.
	\label{eq:iaRZF_ep}	  				
	\end{align}
\end{theorem}
\begin{proof}
	See Appendix~\ref{sec:iaRZF_AsilomarTheo}.
\end{proof}

By employing dominated convergence arguments and the continuous mapping theorem (e.g., \cite{Couillet2011a}), we see that Theorem~\ref{theo:iaRZF_DEofSINR} implies, for each UT$(l,k)$,
\begin{align}
r_{l,k} - \log_2 (1+\overline{\mathrm{SINR}}_{l,k}) \xrightarrow[N \to \infty]{\mathrm{a.s.}} 0 \,.
\end{align}

These results have already been used in Section~\ref{sec:iaRZF_SimpleiaRZF} and will also serve as the basis in the following.

\section{Numerical Results}
\label{sec:iaRZF_Numerics}

In this section we will, first, introduce a heuristic generalization of the previously found (see Paragraph~\ref{sssec:iaRZF_LimitOptimization}) ``limit-optimal'' iaRZF precoder weights.
Furthermore, we provide simulations that corroborate the viability of the proposed precoder, also in systems that are substantially different to the idealized system used in Section~\ref{sec:iaRZF_SimpleiaRZF}.

\subsection{Heuristic Generalization of Optimal Weights}
\label{ssec:iaRZF_HeurisicWeights}
Following from the encouraging sum-rate performance results in Paragraph~\ref{sssec:iaRZF_LimitOptimization}, it seems promising to intuitively generalize the heuristic weights to systems with arbitrary many BSs, transmit powers, CSI randomness and user/antenna ratios.
Adapting the structures obtained in~\eqref{eq:iaRZF_SimpSysbetaoptinf} and~\eqref{eq:iaRZF_SimpSysalphaoptinf}, we define the general heuristic precoder weights as
\begin{align}
\tilde{\alpha}^a_b = \frac{P_a (1-(\tau^a_b)^2)}{P_b c_a \varepsilon^a_b (\tau^a_b)^2 + 1} \,.
\label{eq:iaRZF_heuristicgeneralweights}
\end{align}
Here, we introduced the new notation $\varepsilon^a_b$, which is the average gain factor between BS $a$ and the UTs of cell $b$, i.e., $\varepsilon^a_b = \frac{1}{K_b}\sum_k \chi^a_{b,k}$.
One can intuitively understand \eqref{eq:iaRZF_heuristicgeneralweights} by remembering that $\alpha^a_b$ represents the ``importance'' of the associated channels (from BS$a$ to UTs in cell $b$).
The numerator increases the weight, i.e., makes orthogonality a priority, if the interfering BS uses a large transmit power ($P_a$). Also, importance is lowered for poorly estimated channels.
The denominator reduces orthogonality to cells whose performance is expected to be bad, i.e., $c_b$ approaches $1$, which makes sense from a sum-rate optimization point of view. However, this aspect should be revisited, if interference mitigation is deemed more important than throughput. Weights are also lowered for cells that tolerate interference better due to large own transmit power ($P_b$).
Poor channel estimates reduce importance again. 
The intuitive reason for having $\varepsilon^a_b$ in the denominator is not immediately evident, since one would expect to place lower importance on UTs that are very far away. 
This behavior becomes clear over the realization, that the estimated channels in our model are not normalized (see~\eqref{eq:iaRZF_estimatedchannel}). Thus, the approximate effective weight of the precoder with respect to a normalized channel is $w^a_b=\tilde{\alpha}^a_b \varepsilon^a_b$.
Hence, for $\varepsilon^a_b \to 0$, we have $w^a_b \to 0$, i.e., no importance is placed on very weak channels. Using the same deliberation for $\varepsilon^a_b \to \infty$ we have $w^a_b$ tending to a constant value and for $\tau^a_b \to 0$ we have $w^a_b \to P_a \varepsilon^a_b$. Thus no energy is wasted on far away interferers/weak channels, even if one has perfect CSI of those channels.

We remind ourselves that in order to arrive at~\eqref{eq:iaRZF_heuristicgeneralweights}, we assumed $\xi=1$. Furthermore, systems serving one cluster of closely located UTs per BS, reproduce the initial simplified system closely and, thus, should respond particularly well to the heuristic weights.

\subsection{Performance}
\label{ssec:iaRZF_GenModNumerics}
\begin{figure}
	\centering
	\includegraphics[width=0.3\textwidth]{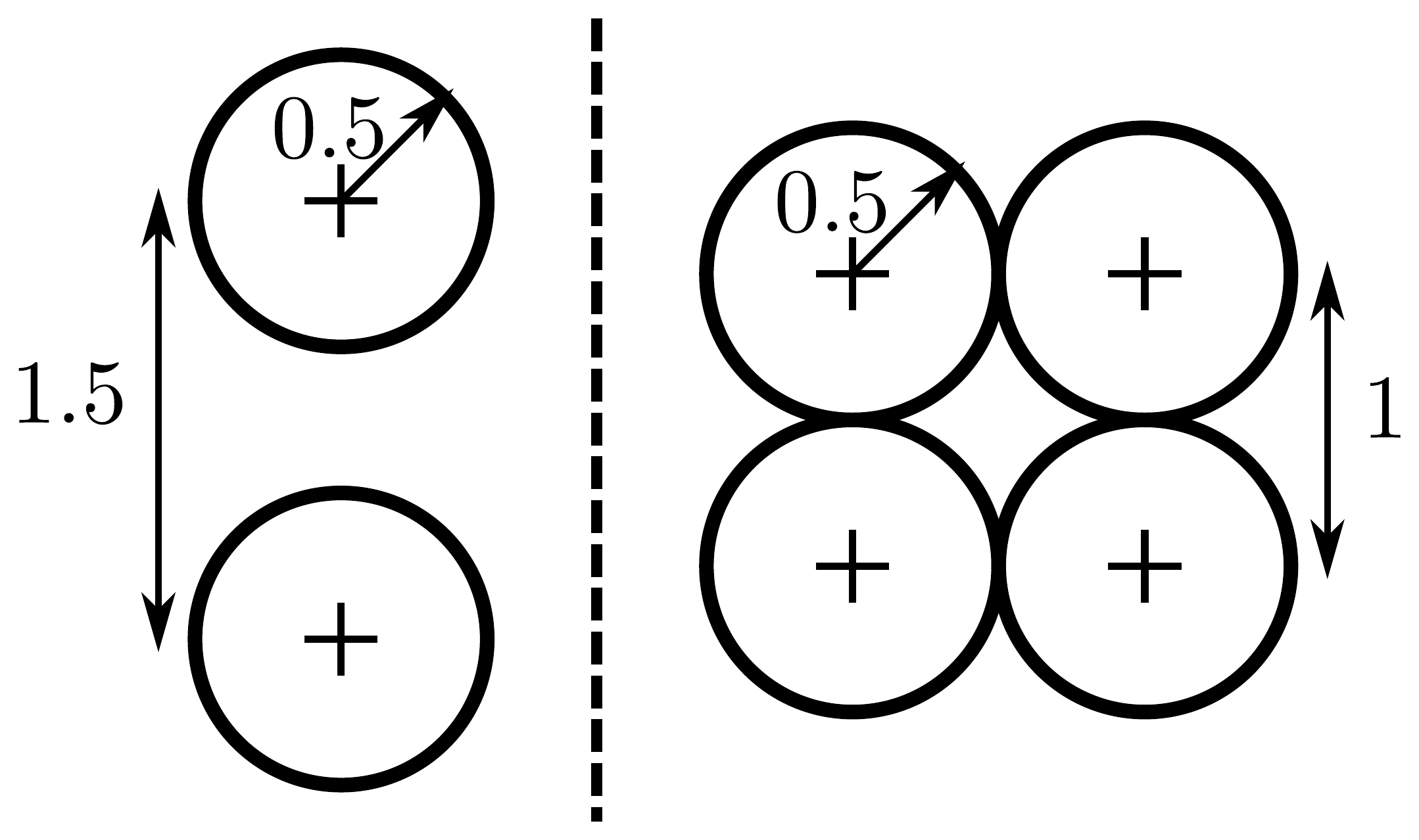}
	\vspace{0.5em}
	\caption{Geometries of the $2$ BS and $4$ BS downlink models.}
	\label{fig:iaRZF_GenMod2BS4BS}
\end{figure}
In order to verify viability of the heuristic approach, we introduce two models (see Figure~\ref{fig:iaRZF_GenMod2BS4BS}). In the first one, two BSs are distanced $1.5$ units, have a height of $0.1$ units and use $160$ antennas each. Around each BS, $40$ single antenna UTs of height $0$, are randomly (uniformly) distributed within a radius of $1$ unit. Hence, one obtains clear non-overlapping clusters that are closely related to the Wyner-like simplified model in Section~\ref{sec:iaRZF_SimpleiaRZF}. The pathloss between each BS and all UTs is defined as the inverse of the distance to the power of $2.8$.
The quality of CSI estimation between a BS and its associated UTs is denoted by $\tau^1_1=\tau^2_2=\tau_a$ and inter cell wise by $\tau^1_2=\tau^2_1=\tau_b$. 
Given the new more general $2$ BS model with randomly distributed UTs for each realization, one has to consider $4$ different channel weights ($\alpha^1_1$, $\alpha^2_2$, $\alpha^1_2$ and $\alpha^2_1$) in the whole system. 
The transmit power to noise ratio (per UT) at each BS is taken to be equal, i.e., $P_1=P_2=P$.
For this system we obtain the average UT rate performance, shown in Figure~\ref{fig:iaRZF_GenSysSimFig1}.
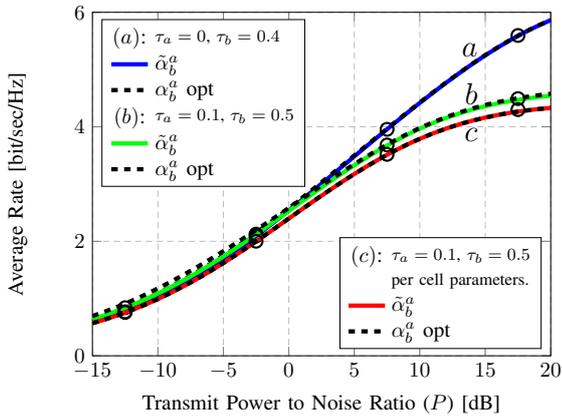
\begin{figure}
	\centering
	\pgfplotsset{
	legend image with text/.style={
		legend image code/.code={%
			\node[anchor=center] at (0.3cm,0cm) {#1};
		}
	},
}

   \begin{tikzpicture}[scale=0.8,font=\normalsize]
    \tikzstyle{every major grid}+=[style=densely dashed]
    \tikzstyle{every axis legend}+=[cells={anchor=west},fill=white,
        at={(0.02,0.98)}, anchor=north west, font=\normalsize ]
    \begin{axis}[
      xmin=-15,
      ymin=0,
      xmax=20,
      ymax=6,
      grid=major,
      scaled ticks=true,
   			xlabel={Transmit Power to Noise Ratio ($P$) [dB]},
   			ylabel={Average Rate [bit/sec/Hz]},
      x post scale=1.1			
      ]
\addlegendimage{legend image with text={ $(a)$:}}
\addlegendentry{\footnotesize  $\tau_a=0$, $\tau_b=0.4$}    
	%
	\addplot[color=blue, no marks, mark size=1.5pt,mark=o,line width=2pt] coordinates{
	(-15.000,0.575) (-14.000,0.646) (-13.000,0.726) (-12.000,0.814) (-11.000,0.910) (-10.000,1.016) (-9.000,1.131) (-8.000,1.256) (-7.000,1.389) (-6.000,1.531) (-5.000,1.682) (-4.000,1.841) (-3.000,2.006) (-2.000,2.178) (-1.000,2.356) (0.000,2.538) (1.000,2.723) (2.000,2.912) (3.000,3.101) (4.000,3.292) (5.000,3.483) (6.000,3.674) (7.000,3.863) (8.000,4.050) (9.000,4.235) (10.000,4.417) (11.000,4.594) (12.000,4.767) (13.000,4.934) (14.000,5.095) (15.000,5.247) (16.000,5.391) (17.000,5.525) (18.000,5.649) (19.000,5.762) (20.000,5.863) 
	};
	\addlegendentry{ {$ \tilde{\alpha}^a_b$} } 
    \addplot[color=black, only marks, mark size=3pt,mark=o,line width=1pt, forget plot] coordinates{
    (-12.500,0.769) (-2.500,2.092) (7.500,3.956) (17.500,5.595) 
    };
    %
	\addplot[color=black, no marks, dashed, mark size=1.5pt,mark=o,line width=2pt] coordinates{
		(-15.000,0.625) (-14.000,0.702) (-13.000,0.787) (-12.000,0.880) (-11.000,0.981) (-10.000,1.089) (-9.000,1.206) (-8.000,1.331) (-7.000,1.464) (-6.000,1.603) (-5.000,1.750) (-4.000,1.903) (-3.000,2.063) (-2.000,2.229) (-1.000,2.400) (0.000,2.576) (1.000,2.756) (2.000,2.939) (3.000,3.124) (4.000,3.312) (5.000,3.499) (6.000,3.687) (7.000,3.874) (8.000,4.060) (9.000,4.243) (10.000,4.424) (11.000,4.600) (12.000,4.772) (13.000,4.939) (14.000,5.099) (15.000,5.251) (16.000,5.394) (17.000,5.529) (18.000,5.652) (19.000,5.765) (20.000,5.867) 
	};
	\addlegendentry{ {$ \alpha^a_b$ opt} }     
	%
    %
    \node at (axis cs:15,5.1) [anchor=south east] { {\Large $a$} };
\addlegendimage{legend image with text={ $(b)$:}}
\addlegendentry{\footnotesize  $\tau_a=0.1$, $\tau_b=0.5$} 
	\addplot[color=green, no marks, mark size=1.5pt,mark=o,line width=2pt] coordinates{
		(-15.000,0.640) (-14.000,0.714) (-13.000,0.795) (-12.000,0.884) (-11.000,0.982) (-10.000,1.087) (-9.000,1.201) (-8.000,1.323) (-7.000,1.453) (-6.000,1.590) (-5.000,1.735) (-4.000,1.886) (-3.000,2.042) (-2.000,2.202) (-1.000,2.366) (0.000,2.530) (1.000,2.695) (2.000,2.859) (3.000,3.020) (4.000,3.176) (5.000,3.328) (6.000,3.472) (7.000,3.609) (8.000,3.738) (9.000,3.857) (10.000,3.966) (11.000,4.065) (12.000,4.153) (13.000,4.231) (14.000,4.300) (15.000,4.359) (16.000,4.410) (17.000,4.452) (18.000,4.488) (19.000,4.518) (20.000,4.543) 
	};
	\addlegendentry{ {$ \tilde{\alpha}^a_b$} }
	%
	\addplot[color=black, no marks, dashed, mark size=1.5pt,mark=o,line width=2pt] coordinates{
		(-15.000,0.694) (-14.000,0.776) (-13.000,0.865) (-12.000,0.961) (-11.000,1.065) (-10.000,1.177) (-9.000,1.295) (-8.000,1.420) (-7.000,1.551) (-6.000,1.687) (-5.000,1.828) (-4.000,1.973) (-3.000,2.121) (-2.000,2.273) (-1.000,2.427) (0.000,2.583) (1.000,2.739) (2.000,2.895) (3.000,3.050) (4.000,3.201) (5.000,3.348) (6.000,3.489) (7.000,3.624) (8.000,3.751) (9.000,3.869) (10.000,3.978) (11.000,4.078) (12.000,4.168) (13.000,4.248) (14.000,4.319) (15.000,4.381) (16.000,4.434) (17.000,4.480) (18.000,4.518) (19.000,4.550) (20.000,4.577) 
	};
	\addlegendentry{ {$ \alpha^a_b$ opt} }   
    \addplot[color=black, only marks, mark size=3pt,mark=o,line width=1pt, forget plot] coordinates{
    (-12.500,0.840) (-2.500,2.125) (7.500,3.682) (17.500,4.492) 
    };
    %
    \node at (axis cs:15,4.9) [anchor=north east] { {\Large $b$} };
    \end{axis}
    \begin{axis}[
    xmin=-15,
    ymin=0,
    xmax=20,
    ymax=6,
    x post scale=1.1,
    axis x line=none,
    axis y line=none,
    legend style={at={(0.98,0.02)}, anchor=south east, font=\normalsize, cells={anchor=west},fill=white}
    ]   
	\addlegendimage{legend image with text={ $(c)$:}}
	\addlegendentry{\footnotesize  $\tau_a=0.1$, $\tau_b=0.5$}  
	\addlegendimage{empty legend}
	\addlegendentry{\footnotesize  per cell parameters.} 	    
	\addplot[color=red, no marks, mark size=1.5pt,mark=o,line width=2pt] coordinates{
		(-15.000,0.569) (-14.000,0.638) (-13.000,0.714) (-12.000,0.798) (-11.000,0.890) (-10.000,0.991) (-9.000,1.101) (-8.000,1.219) (-7.000,1.345) (-6.000,1.479) (-5.000,1.620) (-4.000,1.768) (-3.000,1.920) (-2.000,2.077) (-1.000,2.237) (0.000,2.398) (1.000,2.559) (2.000,2.719) (3.000,2.875) (4.000,3.028) (5.000,3.175) (6.000,3.315) (7.000,3.448) (8.000,3.572) (9.000,3.686) (10.000,3.790) (11.000,3.884) (12.000,3.968) (13.000,4.042) (14.000,4.106) (15.000,4.161) (16.000,4.208) (17.000,4.248) (18.000,4.281) (19.000,4.308) (20.000,4.330) 
	};
	\addlegendentry{ {$ \tilde{\alpha}^a_b$} }
	\addplot[color=black, no marks, dashed, mark size=1.5pt,mark=o,line width=2pt] coordinates{
		(-15.000,0.569) (-14.000,0.638) (-13.000,0.714) (-12.000,0.798) (-11.000,0.890) (-10.000,0.991) (-9.000,1.101) (-8.000,1.219) (-7.000,1.345) (-6.000,1.479) (-5.000,1.620) (-4.000,1.768) (-3.000,1.920) (-2.000,2.077) (-1.000,2.237) (0.000,2.398) (1.000,2.559) (2.000,2.719) (3.000,2.875) (4.000,3.028) (5.000,3.175) (6.000,3.315) (7.000,3.448) (8.000,3.572) (9.000,3.686) (10.000,3.790) (11.000,3.884) (12.000,3.968) (13.000,4.042) (14.000,4.106) (15.000,4.161) (16.000,4.208) (17.000,4.248) (18.000,4.281) (19.000,4.308) (20.000,4.330) 
	};
	\addlegendentry{ {$ \alpha^a_b$ opt} }   
	\addplot[color=black, only marks, mark size=3pt,mark=o,line width=1pt, forget plot] coordinates{
		(-12.500,0.757) (-2.500,2.001) (7.500,3.518) (17.500,4.300) 
	};  
	\node at (axis cs:15,4.1) [anchor=north east] { {\Large $c$} };   
    \end{axis}
  \end{tikzpicture}   
	\caption[2 BSs: Average rate vs.~transmit power to noise ratio.]{2 BSs: Average rate vs.~transmit power to noise ratio ($N_x=160$, $K_x=40$, $P_x=P$, $(\tau_a,\tau_b)\in\{(0,0.4),(0.1,0.5)\}$, i.e., cases $a$, $b$, $c$).} 
	\label{fig:iaRZF_GenSysSimFig1}
\end{figure}
The markers denote results of MC simulations that randomize over UT placement scenarios and channel realizations, when the precoding weights are chosen as in~\eqref{eq:iaRZF_heuristicgeneralweights}.
The main intention for this graph is to compare the performance under heuristic weights ($\tilde{\alpha}^a_b$) and numerically optimal weights ($ \alpha^a_b$ opt), found via $4$D grid search. We observe that the performance of both approaches is virtually the same. 
We also see that even when diverging from the simple system ($\tau_a = 0$, case $a$) by choosing $\tau_a = 0.1$, i.e., case $b$, the heuristic weights still perform practically the same as exhaustive numerical optimization.
The same holds true for the more extreme case $c$, where we chose to modify the numbers of BS antennas, UTs and the CSI quality on a per cell basis ($N_1=160$, $N_2=60$, $K_1=40$, $N_2=20$, $P_x=P$, $\tau_a=0.1$, $\tau_b=0.5$). 

Also, we look at a more complex system of $4$ BSs (see Figure~\ref{fig:iaRZF_GenMod2BS4BS}). The BSs, of height $0.1$ units, are placed on the corners of a square with edge length $1$ units. The UT distribution and pathloss are chosen as before.
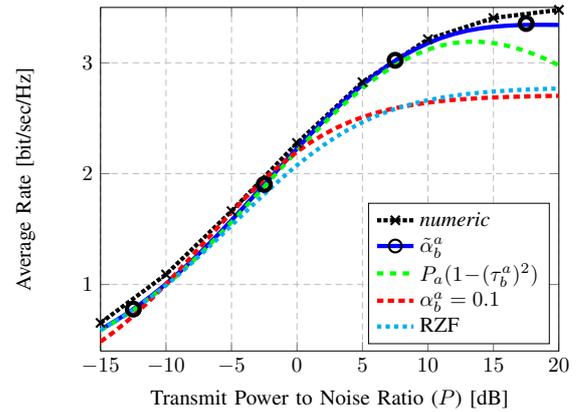
\begin{figure}
	\centering
	   \begin{tikzpicture}[scale=0.8,font=\normalsize]
    \tikzstyle{every major grid}+=[style=densely dashed]
    \tikzstyle{every axis legend}+=[cells={anchor=west},fill=white,
        at={(0.98,0.02)}, anchor=south east, font=\normalsize ]
    \begin{axis}[
      xmin=-15,
      ymin=0.4,
      xmax=20,
      ymax=3.5,
      grid=major,
      scaled ticks=true,
   			xlabel={Transmit Power to Noise Ratio ($P$) [dB]},
   			ylabel={Average Rate [bit/sec/Hz]},
      x post scale=1.1			
      ]
    \addplot[color=black!100, densely dotted, mark size=3pt,mark=x,line width=2pt, mark options={solid, line width=1pt}] coordinates{
    (-15.000,0.653) (-10.000,1.091) (-5.000,1.662) (0.000,2.277) (5.000,2.826) (10.000,3.213) (15.000,3.405) (20.000,3.478)
    };
    \addlegendentry{ {\emph{numeric}} }  
    \addplot[color=black!100, only marks, mark size=3pt,mark=o,line width=2pt, forget plot] coordinates{
    (-12.500,0.778) (-2.500,1.903) (7.500,3.023) (17.500,3.355)
    };
%
    %
    \addplot[color=blue, no marks, mark size=1.5pt,mark=o,line width=2pt, forget plot] coordinates{
    (-15.000,0.593) (-14.000,0.662) (-13.000,0.737) (-12.000,0.819) (-11.000,0.909) (-10.000,1.005) (-9.000,1.109) (-8.000,1.218) (-7.000,1.334) (-6.000,1.454) (-5.000,1.579) (-4.000,1.707) (-3.000,1.837) (-2.000,1.967) (-1.000,2.097) (0.000,2.226) (1.000,2.352) (2.000,2.473) (3.000,2.589) (4.000,2.699) (5.000,2.801) (6.000,2.895) (7.000,2.980) (8.000,3.056) (9.000,3.121) (10.000,3.177) (11.000,3.224) (12.000,3.261) (13.000,3.290) (14.000,3.311) (15.000,3.326) (16.000,3.335) (17.000,3.341) (18.000,3.343) (19.000,3.343) (20.000,3.342)
    };
    \addlegendimage{blue, solid, mark size=3pt,mark=o,line width=2pt, mark options={color=black!100, solid, line width=1pt}}
    \addlegendentry{ {$\tilde{\alpha}^a_b$} }
    %
    \addplot[color=green, dashed, no marks, mark size=1.5pt,mark=o, line width=2pt] coordinates{ 
    (-15.000,0.588) (-14.000,0.656) (-13.000,0.730) (-12.000,0.811) (-11.000,0.900) (-10.000,0.995) (-9.000,1.098) (-8.000,1.206) (-7.000,1.320) (-6.000,1.440) (-5.000,1.563) (-4.000,1.690) (-3.000,1.818) (-2.000,1.948) (-1.000,2.077) (0.000,2.204) (1.000,2.328) (2.000,2.448) (3.000,2.563) (4.000,2.670) (5.000,2.771) (6.000,2.862) (7.000,2.943) (8.000,3.014) (9.000,3.073) (10.000,3.120) (11.000,3.156) (12.000,3.179) (13.000,3.190) (14.000,3.190) (15.000,3.178) (16.000,3.156) (17.000,3.124) (18.000,3.083) (19.000,3.034) (20.000,2.978)  
    };
    \addlegendentry{ {$P_a(1-(\tau^a_b)^2)$} }  
    %
    \addplot[color=red, densely dashed, no marks,  mark size=2.5pt,mark=x,line width=2pt] coordinates{
    (-15.000,0.484) (-14.000,0.570) (-13.000,0.666) (-12.000,0.770) (-11.000,0.883) (-10.000,1.003) (-9.000,1.129) (-8.000,1.259) (-7.000,1.390) (-6.000,1.521) (-5.000,1.650) (-4.000,1.775) (-3.000,1.893) (-2.000,2.004) (-1.000,2.106) (0.000,2.197) (1.000,2.279) (2.000,2.351) (3.000,2.413) (4.000,2.465) (5.000,2.510) (6.000,2.547) (7.000,2.578) (8.000,2.603) (9.000,2.624) (10.000,2.641) (11.000,2.655) (12.000,2.666) (13.000,2.674) (14.000,2.682) (15.000,2.687) (16.000,2.692) (17.000,2.695) (18.000,2.698) (19.000,2.701) (20.000,2.702) 
    };
    \addlegendentry{ {$\alpha^a_b = 0.1$} }	
    %
    \addplot[color=cyan, dotted, no marks, mark size=1.5pt,mark=o,line width=2pt] coordinates{
    (-15.000,0.590) (-14.000,0.658) (-13.000,0.732) (-12.000,0.812) (-11.000,0.900) (-10.000,0.993) (-9.000,1.093) (-8.000,1.198) (-7.000,1.307) (-6.000,1.419) (-5.000,1.533) (-4.000,1.647) (-3.000,1.760) (-2.000,1.870) (-1.000,1.976) (0.000,2.076) (1.000,2.169) (2.000,2.255) (3.000,2.333) (4.000,2.403) (5.000,2.464) (6.000,2.518) (7.000,2.564) (8.000,2.603) (9.000,2.635) (10.000,2.663) (11.000,2.685) (12.000,2.704) (13.000,2.719) (14.000,2.731) (15.000,2.741) (16.000,2.749) (17.000,2.756) (18.000,2.761) (19.000,2.765) (20.000,2.768)
    };
    \addlegendentry{ {RZF} }     	
    \end{axis}
  \end{tikzpicture}   
	\caption[4 BSs: Average rate vs.~transmit power to noise ratio.]{4 BSs: Average rate vs.~transmit power to noise ratio ($N_x=160$, $K_x=40$, $P_x=P$ , $(\tau_a,\tau_b,\tau_c)=(0.1,0.3,0.4)$).} 
	\label{fig:iaRZF_GenSysSimFig2}
\end{figure}
Figure~\ref{fig:iaRZF_GenSysSimFig2} shows the performance of the $4$ BS system, assuming that each BS has $160$ antennas with a power constraint of $P$ per UT and serves $40$ UTs. We assume that the CSI randomness is overwhelmingly determined by inter-BS distance, i.e., we have $\tau_a$ for each BS to the adherent UTs, $\tau_b$ for each BS to UTs of BSs $1$ unit away and $\tau_c$ for each BS to UTs of BSs $\sqrt{2}$ units away. It is, thus, reasonable to choose $\tau_a<\tau_b<\tau_c$.
In the graph we compare the heuristic weights ($\tilde{\alpha}^a_b$) with various other weighting approaches. Round markers stem from Monte-Carlo simulations of the performance pertaining to the heuristic weights DEs and confirm the viability of the large scale approximation.
The benchmark \emph{numeric} result in this figure is obtained from optimizing the $16$ precoder weights via extensive numerical search, using $\tilde{\alpha}^a_b$ as a starting point. 
Of practical interest is the performance of the simplified heuristics $\alpha^a_b = P_a(1-(\tau^a_b)^2)$. This choice means that no interference is taken into account, i.e., $\varepsilon^a_b = 0$. We observe that most of the viability of the heuristic method comes from this part. Only at low SNR, where interference is the dominant problem, and very high SNR does the $P_a(1-(\tau^a_b)^2)$ approach become noticeably suboptimal. 
The constant weight approach ($\alpha^a_b = 0.1$) behaves like in Section~\ref{sec:iaRZF_SimpleiaRZF}, in that it is only a good match  for a limited part of the curve. 
For comparison purposes, we also compare with standard non-cooperative RZF, as defined in Subsection~\ref{ssec:iaRZF_PerfSimpleSystem}.

In general, employing $\tilde{\alpha}^a_b$ is most advantageous in high interference scenarios, as would be expected due to the ``interference aware'' conception of the precoder.
Carrying out the same simulations for different levels of CSI randomness, one observes that the gain of using the heuristic variant of iaRZF is substantial as long as the estimations of the interfering channels are not too bad. For extremely bad CSI, standard non-cooperative RZF can outperform iaRZF with $\tilde{\alpha}^a_b$. 
We also note that better CSI widens the gap between the $\tilde{\alpha}^a_b$ and $\alpha^a_b = P_a(1-(\tau^a_b)^2)$ weighted iaRZF approaches.

\section{Conclusion}
In this paper, we analyzed a linear precoder structure for multi cell systems, based on an intuitive interference induction trade-off and recent results on multi cell RZF, denoted iaRZF.
It was shown that the relegation of interference into orthogonal subspaces by iaRZF can be explained rigorously and intuitively, even without assuming large scale systems.
For example, one can indeed observe that the precoder can either completely get rid of inter cell or intra cell interference (assuming perfect channel knowledge).

Stating and proving new results from large-scale random matrix theory, allowed us to give more conclusive and intuitive insights into the behavior of the precoder, especially with respect to imperfect CSI knowledge and induced interference mitigation. The effectiveness of these large-scale results has been demonstrated in practical finite dimensional systems.
Most importantly, we concluded that iaRZF can use all available (also very bad) interference channel knowledge to obtain significant performance gains, while not requiring explicit inter base station cooperation.

Moreover, it is possible to analytically optimize the iaRZF precoder weights in certain limit scenarios using our large-scale results.
Insights from this were used to propose a heuristic generalization of the limit optimal iaRZF weighting for arbitrary systems. 
The efficacy of the heuristic iaRZF approach has been demonstrated by achieving a sum-rate close to the numerically optimally weighted iaRZF, for a wide range of general and practical systems.
The effectiveness of our heuristic approach has been intuitively explained by mainly balancing the importance of available knowledge about various channel and system variables.

\appendices

\section{Useful Notation and Lemmas}
\label{sec:iaRZF_LemmasTools}
In this appendix we give some frequently used lemmas and definitions to facilitate exposition in the following.

\begin{lemma}[Common Matrix Identities]\label{lem:iaRZF_MatIDs}
	Let $\Am$, $\Bm$ be complex invertible matrices and $\Cm$ a  rectangular complex matrix, all of proper size. We restate the following, well known, relationships:\\
	Woodbury Identity:
	\begin{align}
	& \lr \Am +\Cm\Bm\Cm^\H \rr^\mo = \nonumber \\ 
	&\quad \Am^\mo - \Am^\mo\Cm\lr\Bm^\mo+\Cm^\H\Am^\mo\Cm\rr^\mo\Cm^\H\Am^\mo \label{eq:iaRZF_WoodI}.
	\end{align}
	Searl Identity:
	\begin{align}
	\lr \Id + \Am\Bm \rr^\mo \Am = \Am \lr \Id + \Bm\Am \rr^\mo \label{eq:iaRZF_SearlI}.
	\end{align}
	Resolvent Identity:
	\begin{align}
	\Am^\mo + \Bm^\mo = -\Am^\mo \lr \Am-\Bm \rr \Bm^\mo \label{eq:iaRZF_ResolventId}.
	\end{align}
\end{lemma}

\begin{lemma}[Unitary Projection Matrices]\label{lem:iaRZF_Projection}
	Let $\Xm$ be an $N\times K$ complex matrix, where $N\geq K$ and $\rank(\Xm)=K$. We define $\Pm_\Xm = \Xm\lr\Xm^\H\Xm\rr^\mo\Xm^\H$ and $\Pm_\Xm^\perp= \Id - \Pm_\Xm$.
	It follows (see e.g., \cite[Chapter~5.13]{Meyer2000}) 
	\begin{align*}
	\Pm = \Pm^2 & \Leftrightarrow \Pm = \Pm^\H \\
	\Pm_\Xm^\perp \Xm = 0 & \Leftrightarrow \Xm^\H \Pm_\Xm^\perp = 0 \,.
	\end{align*}
	Generally one denotes $\Pm_\Xm$ as the projection matrix onto the column space of $\Xm$ and $\Pm_\Xm^\perp$ as the projection matrix onto the orthogonal space of the column space of $\Xm$.
\end{lemma}

\begin{definition}[Notation of Resolvents] \label{def:iaRZF_DefResolvents}
	Given the notations from Section~\ref{sec:iaRZF_Asilomar}, we define resolvent matrices of $\hat{\Hm}^a_{a}$ as:
	\begin{align*}
	\Qm_{a} &\eqdef \lr \alpha^a_a \hat{\Hm}^a_{a} (\hat{\Hm}^a_{a})^\H + \Zm^a + \xi_a \Id_{N_a} \rr^\mo \,.
	\end{align*}
	We will also use of the following modified versions 
	\begin{align*}
	\Qm_{a[bc]} &\eqdef \lr \alpha^a_a \hat{\Hm}^a_{a} (\hat{\Hm}^a_{a})^\H + \Zm^a - \alpha^a_b \hat{\hv}^{a}_{b,c} (\hat{\hv}^{a}_{b,c})^\H + \xi_a \Id_{N_a} \rr^\mo \\
	\Qm_{a[b]} &\eqdef \lr \alpha^a_a \hat{\Hm}^a_{a} (\hat{\Hm}^a_{a})^\H + \Zm^a - \alpha^a_a \hat{\hv}^{a}_{a,b} (\hat{\hv}^{a}_{a,b})^\H + \xi_a \Id_{N_a} \rr^\mo \nonumber \\
	&= \lr \alpha^a_a \hat{\Hm}^a_{a[b]} (\hat{\Hm}^a_{a[b]})^\H + \Zm^a + \xi_a \Id_{N_a} \rr^\mo \,.
	\end{align*}
\end{definition}

\begin{lemma}[{Matrix Inversion Lemma \cite[Lemma~2.2]{SilversteinBai1995}}]\label{lem:iaRZF_MILI}
	Let $\Am$ be an $M\times M$ invertible matrix and $\xv \in \CC^M, c \in \CC$ for which $\Am +c \xv\xv^\H$ is invertible. Then, as an application of \eqref{eq:iaRZF_WoodI}, we have
	\begin{align*}
	\xv^\H\lr\Am+c\xv\xv^\H\rr^\mo = \frac{\xv^\H\Am^\mo}{1+c\xv^\H\Am^\mo\xv}.
	\end{align*}
	For the previously defined resolvent matrices, we have in particular
	\begin{align*}
	\Qm_a \hat{\hv}^a_{a,b} = \frac{\Qm_{a[b]} \hat{\hv}^a_{a,b}}{1 + \alpha^a_a (\hat{\hv}^a_{a,b})^\H \Qm_{a[b]} \hat{\hv}^a_{a,b} }.
	\end{align*}
\end{lemma}

\begin{lemma}[Convergence of Quadratic Forms \cite{SIL98}]	\label{lem:iaRZF_quadratic} \label{lem:iaRZF_TrL}
Let ${\bf x}_M=\left[X_1,\ldots,X_M\right]^{\mbox{\tiny T}}$ be an $M\times 1$ vector where the $X_n$ are i.i.d.\ Gaussian complex random variables with unit variance. Let ${\bf A}_M$ be an $M\times M$ matrix independent of ${\bf x}_M$.
If in addition $ \limsup_M \|{\bf A}\|_2 <\infty $ then we have that
$$
\frac{1}{M}{\bf x}_M^{\mbox{\tiny H}}{\bf A}_M{\bf x}_M-\frac{1}{M} \tr ({\bf A}_M) \xrightarrow[M\to+\infty]{a.s.} 0 \,.
$$
\end{lemma}

\begin{cor}
	\label{corollary:iaRZF_zero_quadratic} \label{cor:iaRZF_TrL0} 
	Let ${\bf A}_M$ be as in Lemma~\ref{lem:iaRZF_quadratic}, i.e., $ \limsup_M \|{\bf A}\|_2 <\infty $, and ${\bf x}_M,{\bf y}_M$ be random, mutually independent with complex Gaussian entries of zero mean and variance $1$. Then we have 
	$$
	\frac{1}{M} {\bf y}_M^{\mbox{\tiny H}}{\bf A}_M{\bf x}_M\xrightarrow[M\to+\infty]{a.s.}0 \,.
	$$
\end{cor}

\begin{lemma}{[Rank-One Perturbation Lemma \cite[Lemma~14.3]{Couillet2011a}]}
	\label{lem:iaRZF_perturbation} \label{lem:iaRZF_R1PL}
	Let $\Qm_{a}$ and $\Qm_{a[b]}$ be the resolvent matrices as defined in Definition~\ref{def:iaRZF_DefResolvents}. Then, for any matrix ${\bf A}$ we have:
	$$
	\tr \ls {\bf A}\left(\Qm_{a}-\Qm_{a[b]}\right) \rs \leq \frac{1}{\xi_a}\|{\bf A}\|_2 \,.
	$$
\end{lemma}

\section{Simple System Limit Behavior Proofs}
\label{sec:iaRZF_LimitBehaviour}
In this section, we provide the proofs pertaining to the limit behavior of the simple system in Section~\ref{sec:iaRZF_SimpleiaRZF}.

\subsection{Finite Dimensions}
\label{ssec:iaRZF_ProofFinitDim}
In order to simplify the notation we will not explicitly state the index $x$ in the following, unless needed, hence the normalized precoder $\Fm$ for each of the two cells is
$
\Fm = \sqrt{K}{\Mm}/{\sqrt{\tr \Mm\Mm^\H}}
$
for
$\Mm = \lr \alpha\Hm\Hm^\H + \beta\Gm\Gm^\H+\xi\Id \rr^\mo \Hm $.

\subsubsection{$\boldsymbol{\beta\to\infty}$}
For the limit when $\beta\to\infty$ we use~\eqref{eq:iaRZF_WoodI} with $\Am=\beta\Gm\Gm^\H+\xi\Id$ and $\Cm\Bm\Cm^\H=\Hm\alpha\Id\Hm^\H$ to reformulate the matrix $\Mm$
\begin{align*}
\Mm &= \lr \alpha\Hm\Hm^\H + \beta\Gm\Gm^\H+\xi\Id \rr^\mo \Hm \\
&= \ls \Qm_\Gm - \Qm_\Gm\Hm\lr \alpha^\mo\Id+\Hm^\H\Qm_\Gm\Hm \rr^\mo\Hm^\H\Qm_\Gm \rs \Hm
\end{align*}
where
\begin{align*}
\Qm_\Gm &= \lr \beta\Gm\Gm^\H+\xi\Id \rr^\mo \\
&\eqtextt{\eqref{eq:iaRZF_WoodI}} \xi^\mo\Id -\xi^\mo\Gm\lr \frac{\xi}{\beta}\Id+\Gm^\H\Gm \rr^\mo\Gm^\H \,.
\end{align*}
We now let $\beta\to\infty$, assuming $\Gm^\H\Gm$ is invertible (which is true with probability $1$) and $\xi$ bounded.
In this regime, we remember Lemma~\ref{lem:iaRZF_Projection}, and rewrite $\Qm_\Gm = \xi^\mo \Pm_\Gm^\perp$. One finally arrives at
\begin{align*}
\Mm & \betato \\   
    & \ls\xi^\mo\Pm_\Gm^\perp-\xi^\mt\Pm_\Gm^\perp\Hm\lr\alpha^\mo\Id+\xi^\mo\Hm^\H\Pm_\Gm^\perp\Hm\rr^\mo\Hm^\H\Pm_\Gm^\perp\rs\Hm \,.
\end{align*}
Relying further on properties of projection matrices ($\Pm_\Gm^\perp=\Pm_\Gm^\perp\Pm_\Gm^\perp$, $(\Pm_\Gm^\perp)^\H=\Pm_\Gm^\perp$) and introducing the matrix $\check{\Hm}=\Pm_\Gm^\perp\Hm$, as the channel matrix $\Hm$ projected on the space orthogonal to the channels of $\Gm$, we get
\begin{align*}
\Mm & \betato \\ \phantom{=} &\xi^\mo\ls\Pm_\Gm^\perp\Hm-\Pm_\Gm^\perp\Hm\lr\frac{\xi}{\alpha}\Id+\Hm^\H\Pm_\Gm^\perp\Pm_\Gm^\perp\Hm\rr^\mo\Hm^\H\Pm_\Gm^\perp\Pm_\Gm^\perp\Hm\rs \\
= & \xi^\mo\ls\check{\Hm}-\check{\Hm}\lr\Id - \frac{\xi}{\alpha} \lr \frac{\xi}{\alpha}\Id+\check{\Hm}^\H\check{\Hm}\rr^\mo\rr\rs \\
= & \check{\Hm}\lr \xi\Id+\alpha\check{\Hm}^\H\check{\Hm}\rr^\mo \,.
\end{align*}

\subsubsection{$\boldsymbol{\alpha\to\infty}$}
Introducing the abbreviations $\Qm_\Hm = \lr\Hm\Hm^\H+\frac{\xi}{\alpha}\Id\rr^\mo$ and $\bar{\Qm}_\Hm = \lr\Hm^\H\Hm+\frac{\xi}{\alpha}\Id\rr^\mo$, we can rewrite the matrix $\Mm$ as follows.
\begin{align*}
\alpha \Mm & = \lr \Hm\Hm^\H + \frac{\beta}{\alpha}\Gm\Gm^\H+\frac{\xi}{\alpha}\Id \rr^\mo \Hm \\
& \eqtextt{\eqref{eq:iaRZF_WoodI}} 
\ls \Qm_\Hm - \Qm_\Hm\Gm\lr\frac{\alpha}{\beta}\Id+\Gm^\H\Qm_\Hm\Gm\rr^\mo\Gm^\H\Qm_\Hm \rs \Hm \\
& \eqtextt{\eqref{eq:iaRZF_SearlI}} 	
\Hm\bar{\Qm}_\Hm - \Qm_\Hm\Gm\lr\frac{\alpha}{\beta}\Id+\Gm^\H\Qm_\Hm\Gm\rr^\mo\Gm^\H\Hm\bar{\Qm}_\Hm \,. 
\end{align*}
Applying \eqref{eq:iaRZF_ResolventId} to the expression $\lr\Hm\Hm^\H+\frac{\xi}{\alpha}\Id\rr^\mo + \lr-\frac{\xi}{\alpha}\Id\rr^\mo$, one eventually finds the relationship
$ \Qm_\Hm = \alpha\xi^\mo\lr\Id - \Hm\bar{\Qm}_\Hm\Hm^\H \rr $.  Hence,
\begin{align*}
\alpha \Mm  = &
\Hm\bar{\Qm}_\Hm - \xi^\mo\lr\Id - \Hm\bar{\Qm}_\Hm\Hm^\H \rr \\
&\times\Gm\ls\frac{1}{\beta}\Id+\xi^\mo\Gm^\H\lr\Id - \Hm\bar{\Qm}_\Hm\Hm^\H \rr\Gm\rs^\mo\Gm^\H\Hm\bar{\Qm}_\Hm	\,.
\end{align*}
Now, taking the limit of $\alpha\to\infty$, assuming $\Hm^\H\Hm$ invertible (true with probability $1$), and recognizing $\Pm_\Hm^\perp = \Id - \Hm\lr\Hm^\H\Hm\rr^\mo\Hm^\H$ we arrive at
\begin{align*}
& \alpha \Mm  \alphato
\Hm\lr\Hm^\H\Hm\rr^\mo -\gamma^\mo\ls\Id-\Hm\lr\Hm^\H\Hm\rr^\mo\Hm^\H\rs\Gm \\
&\qquad \lc\beta^\mo\Id+\gamma^\mo\Gm^\H\ls\Id-\Hm\lr\Hm^\H\Hm\rr^\mo\Hm^\H\rs\Gm\rc^\mo \\
&\qquad \Gm^\H\Hm\lr\Hm^\H\Hm\rr^\mo	\\
& = 
\Hm\lr\Hm^\H\Hm\rr^\mo \\ & \qquad-\gamma^\mo\Pm_\Hm^\perp\Gm\lc\beta^\mo\Id+\gamma^\mo\Gm^\H\Pm_\Hm^\perp\Gm\rc^\mo\Gm^\H\Hm\lr\Hm^\H\Hm\rr^\mo \,.
\end{align*}

\subsection{Large-Scale Approximation}
\label{ssec:iaRZF_ProofLargeScale}
In this subsection, we primarily show that the fixed point equation $e$ is bounded in the sense of $0 < \liminf e < \limsup e < \infty$. This knowledge simplifies the limit calculations in Subsection~\ref{ssec:iaRZF_LimitBehavioriaRZF} to simple operations.
We remind ourselves, that for perfect and imperfect CSI the resulting fixed point equations are equivalent: 
\begin{align}
e &= \lr 1+ \frac{c}{\alpha^\mo+e} + \frac{c\varepsilon}{\beta^\mo+\varepsilon e} \rr^\mo \label{eq:iaRZF_eEquivalentVersion}
\end{align}
where we abbreviated $e_{\alpha}$ with $e$ for notational convenience.

\begin{lemma}[$e$ is Bounded]\label{lem:iaRZF_ebound}
	For either $\alpha\to\infty$ and $\beta,\varepsilon$ bounded or $\beta\to\infty$ and $\alpha,\varepsilon$ bounded, we have
	$$
	0 < \liminf e < \limsup e < \infty\,.
	$$
\end{lemma}
\begin{proof}
	1) $e < \infty$ when $\alpha$ or $\beta \to\infty$, follows immediately from contradiction, when one takes $e\to\infty$ in~\eqref{eq:iaRZF_eEquivalentVersion}. 
	
	2) To see that $e$ positive when $\alpha$ or $\beta \to\infty$, we take either $\alpha\to\infty$ and $\beta,\varepsilon$ bounded or $\beta\to\infty$ and $\alpha,\varepsilon$ bounded. 
	For the case $\alpha \to \infty$, we first denote $\upsilon = \alpha e$
	and we look at 
	\begin{align*}
	\upsilon &= \lr \frac{1}{\alpha}+ \frac{c}{1+\upsilon} + \frac{c\beta \varepsilon}{\alpha+\beta\varepsilon\upsilon} \rr^\mo \,.
	\end{align*}
	Now we assume $\upsilon$ to be bounded for $\alpha\to\infty$
	\begin{align*}
	\upsilon &= \lim_{\alpha\to\infty} \lr \frac{1}{\alpha}+ \frac{c}{1+\upsilon} + \frac{c\beta \varepsilon}{\alpha+\beta\varepsilon\upsilon} \rr^\mo 
	= \lr\frac{c}{1+\upsilon}\rr^\mo
	\end{align*}
	thus implying $\upsilon = \frac{1}{c-1} <0$, as $c<1$.
	Case 1 directly contradicts the assumption and case 2 is contradicting, as $e$ can not be negative for positive values of $\alpha$, $\beta$, $c$ and $\varepsilon$.
	Thus, $\upsilon$ is not bounded for $\alpha\to\infty$, hence $e$ can neither be zero nor negative.
	For the case of $\beta \to \infty$, we denote $\upsilon = \beta e$ and proceed analogously.
\end{proof}

\subsection{Large-Scale Optimization $\alpha\to\infty$}
\label{ssec:iaRZF_appProofSimpleOptimizationalphainfty}

Continuing from Appendix~\ref{ssec:iaRZF_ProofLargeScale}, we see that in the limit $\alpha\to\infty$ the large-scale approximation of the SINR values pertaining to the users of each cell, i.e., $\overline{\mathrm{SINR}}^{\alpha\to\infty}$, is indeed as stated in Paragraph~\ref{sssec:iaRZF_LimitOptimization}.

Differentiating $\overline{\mathrm{SINR}}^{\alpha\to\infty}$ w.r.t.\ $\beta$, while taking into account that $e$ is an abbreviation for $e^{\alpha\to\infty}_\beta$ leads us to
\begin{align*}
&\frac{\partial \overline{\mathrm{SINR}}^{\alpha\to\infty}}{\partial \beta}  = 
- 2 P c\varepsilon^2
\ls e+\beta e^\prime \rs \numberthis \label{eq:iaRZF_SimpleModelLimitSINRdiff} \\
& \times \frac{ 
	t_1
}
{
	\ls P \lr c\beta^2 e^2 \varepsilon^3\tau^2+2c\beta e \varepsilon^2\tau^2+c\varepsilon\rr +\beta^2e^2\varepsilon^2+2\beta e \varepsilon+1 \rs^2
}  
\end{align*}
where we used $e^\prime$ as shorthand for $\frac{\partial e^{\alpha\to\infty}(\beta)}{\partial \beta}$ and
\begin{align*}
t_1
&= P \ls c-1-\beta\varepsilon e+2\beta c\varepsilon e \rs +\beta e + \beta^2 e^2\varepsilon  \\
&\qquad - P_{\bar{x}} \tau^2 \ls c - 1 -\beta\varepsilon e+\beta c\varepsilon e- \beta^2 c e^2\varepsilon^2 \rs \,.
\end{align*}

Realizing that the denominator of \eqref{eq:iaRZF_SimpleModelLimitSINRdiff} can not become zero, we have two possible solutions for $ \partial \overline{\mathrm{SINR}}^{\alpha\to\infty} / \partial \beta = 0$.
In Lemma~\ref{lem:iaRZF_eplusbep} we show that $ e+\beta e^\prime  > 0$, hence we only need to deal with the term $t_1$.
We remember from \eqref{eq:iaRZF_SimpleModelLimitFPEe} that
$
c-1-\beta\varepsilon e+2\beta c\varepsilon e+e+\beta\varepsilon e^2 = 0 \,.
$
Thus,
\begin{align*}
& c - 1 -\beta\varepsilon e+\beta c\varepsilon e- \beta^2 c e^2\varepsilon^2 
= -\beta c \varepsilon e - e - \beta\varepsilon e^2 - \beta^2 c e^2\varepsilon^2
\end{align*}
and similarly
\begin{align*}
& P \ls c-1-\beta\varepsilon e+2\beta c\varepsilon e \rs +\beta e + \beta^2 e^2\varepsilon \\
&\qquad = -P e - P\beta\varepsilon e^2 + \beta e + \beta^2\varepsilon e^2 \,. 
\end{align*}
Hence, 
\begin{align*}
t_1 
& = \lr \varepsilon e^2 +  P\tau^2 c e^2\varepsilon^2\rr
\lr \beta - \frac{P (1-\tau^2) }{P c\varepsilon\tau^2+1} \rr
\lr \beta + \frac{1}{e \varepsilon} \rr \,.
\end{align*}
Given that only the middle term can become zero, we find $\beta_{opt}$ to be
\begin{align}
\beta_{opt} = \frac{P (1-\tau^2) }{P c\varepsilon\tau^2+1}
\end{align}
as stated in \eqref{eq:iaRZF_SimpSysbetaoptinf}.
The physical interpretation of the SINR guarantees this point to be the maximum.


We used the assumption $e+\beta e^\prime > 0$ to arrive at the previous result. This claim is proved by the following lemma.
\begin{lemma}\label{lem:iaRZF_eplusbep}
	Given the notation and definitions from Appendix~\ref{ssec:iaRZF_appProofSimpleOptimizationalphainfty}, we have that $e+\beta e^\prime > 0$. 
\end{lemma}
\begin{proof}[Proof Sketch]
From \cite{SilversteinBai1995} we know that an object of the form
\begin{align*}
	m(z) = \big[ -z + c\int \frac{t}{1+tm(z)}d\nu(t) \big]^\mo 
\end{align*}
where $\nu$ is a non negative finite measure, is a so-called Stieltjes transform of a measure $\nu$, defined $\forall z \notin \supp(\nu)$.
Adapting \eqref{eq:iaRZF_eEquivalentVersion} by re-naming $\tilde{e} \eqdef \beta\varepsilon e$ we see that it is indeed a valid Stieltjes transform for an appropriately chosen measure. Finally, one recognizes $\beta e^\prime + e$ as the derivative of a Stieltjes transform, which is always positive.
\end{proof}

\section{Proof of Theorem~\ref{theo:iaRZF_DEofSINR}}
\label{sec:iaRZF_AsilomarTheo}
The objective of this section is to find a DE for the SINR term~\eqref{eq:iaRZF_FDimSINR}. A broad outline of the required steps is as follows. 
In the beginning of the proof we condition that $\Zm^m$ is fixed to some realization and we follow the steps given in \cite[Appendix~II]{Wagner2012a} for the power normalization $\nu_m$. Invoking \cite[Theorem~1]{Wagner2012a} we obtain the fundamental equations for $e_{m}$. We, then, allow $\Zm^m$ to be random and apply \cite[Theorem~3.13]{Couillet2011a} to obtain \eqref{eq:iaRZF_e}. Invoking Tonelli's theorem, it is admissible to apply the two theorems one after the other, as $\Zm^m$ is a bounded sequence with probability one. 
The DEs of all required terms are found by following \cite[Appendix~II]{Wagner2012a} again. 
This is true for the terms from Subsection~\ref{ssec:iaRZF_SINRfiniteDim}, as well. However here the interference terms ask for a slightly more generalized version of \cite[Lemma~7]{Wagner2012a}.

\subsection{Power Normalization Term}
\label{ssec:iaRZF_appAsilomarPwrNorm} 

We start by finding a DE of the term $\nu_m$, which will turn out to be a frequently reoccurring object throughout this Section. 
From \eqref{eq:iaRZF_FDimPwrCtrl}, we see that the power normalization term $\nu_m$ is defined by the relationship
\begin{align*}
\frac{P_m}{\nu_m} \frac{K_m}{N_m}
= & \frac{1}{N_m}\trace\ls  \hat{\Hm}^m_m (\hat{\Hm}^m_m)^\H \Qm_m^2 \rs  \\
= & \frac{\partial}{\partial \xi_m}\lc \frac{1}{\alpha^m_m N_m}\trace\ls 	\lr\Zm^m + \xi_m \Id_{N_m}\rr \Qm_m \rs  \rc \numberthis \label{eq:iaRZF_pwrnormdiffreform}
\end{align*}
where we used the general identities 
$
\frac{\partial}{\partial y}\lc -\trace\ls \Am\lr\Am+\Bm+y\Id\rr^\mo \rs\rc = \trace\ls \Am \lr\Am+\Bm+y\Id\rr^\mt \rs 
$
and
$
\Am \lr\Am+\Bm+y\Id\rr^\mo = \Id - \lr\Bm+y\Id\rr \lr\Am+\Bm+y\Id\rr^\mo \,.
$
The goal now is to find a deterministic object $\bar{X}_m$ that satisfies
\begin{align*}
\frac{1}{N_m}\trace\ls  \hat{\Hm}^m_m (\hat{\Hm}^m_m)^\H \Qm_m^2 \rs - \bar{X}_m \xrightarrow[N \to \infty]{\mathrm{a.s.}} 0
\end{align*}
for the regime defined in A-\ref{as:iaRZF_cchain}.

To do this, we apply \cite[Theorem~1]{Wagner2012a} to~\eqref{eq:iaRZF_pwrnormdiffreform}, where we set the respective variables to be $\Psim_i = \chi^m_{m,i}\Id$, $\Qm_N = \Zm^m + \xi_m \Id_{N_m}$, $\Bm_N = \alpha^m_m \hat{\Hm}^m_m (\hat{\Hm}^m_m)^\H + \Zm^m$ and $z = -\xi_m$.
Thus, we find the (partially deterministic) quantity
\begin{align*}
\bar{X}_m =  \frac{\partial}{\partial \xi_m} & \frac{1}{\alpha^m_m N_m}\trace\bigg[ 	
\lr\Zm^m + \xi_m \Id_{N_m}\rr  \\
    &  \bigg( \frac{1}{N_m} \sum_{j=1}^{K_m} \frac{\alpha^m_m\chi^m_{m,j}\Id_{N_m}}{1+e^j_{m}} + \Zm^m + \xi_m \Id_{N_m} \bigg)^\mo 
\bigg] 
\end{align*}
where $e^j_{m} = \alpha^m_m\chi^m_{m,j} e_{m}$ and
\begin{align*}
e_{m} = 
\frac{1}{N_m}\trace \lr \frac{1}{N_m} \sum_{j=1}^{K_m} \frac{\alpha^m_m\chi^m_{m,j}\Id_{N_m}}{1+\alpha^m_m\chi^m_{m,j} e_{m}} + \Zm^m + \xi_m \Id_{N_m} \rr^\mo \,.
\end{align*}

\begin{remark} \label{rem:iaRZF_simpleDEofe}
	In order to reuse the results from this section later on, it will turn out to be useful to realize the following relationship involving $e_{m}$.
	\begin{align}
	\frac{1}{N_m}\trace\Qm_m - e_{m} \xrightarrow[N \to \infty]{\mathrm{a.s.}} 0 \,. \label{eq:iaRZF_SimpleEquationtoe} 
	\end{align}
	This can be quickly verified by using \cite[Theorem~1]{Wagner2012a}, via choosing $\tilde{\Rm}_i = \chi^m_{m,i}\Id$, $\Dm_N = \Id$, $\Bm_N =\alpha^m_m \hat{\Hm}^m_m (\hat{\Hm}^m_m)^\H + \Zm^m$ and $z = -\xi_m$.
\end{remark}

One notices, that the fixed-point equation $e_{m}$ contains the term $\Zm^m$, which is not deterministic. Thus, the derived objects are not yet DEs.
In order to resolve this we need to condition $\Zm^m$ to be fixed, for now.
Under this assumption we now find the DE of $e_{m}$. To do this, it is necessary to realize that $e_{m}$ contains another Stieltjes transform:
\begin{align*}
e_{m} = \frac{1}{N_m}\trace\ls \lr \Zm^m + \beta_m \Id_{N_m} \rr^\mo  \rs
\end{align*}
where 
\begin{align}
\beta_m = \frac{1}{N_m} \sum_{j=1}^{K_m} \frac{\alpha^m_m\chi^m_{m,j}}{1+\alpha^m_m\chi^m_{m,j} e_{m} } + \xi_m \,. \label{eq:iaRZF_betaforsecondST}
\end{align}
The solution becomes immediate once we rephrase $\Zm^m$ as
\begin{align*}
\Zm^m 	& = \sum_{l\neq m} \sum_{k=1}^{K_l} \alpha^m_l \hat{\hv}^{m}_{l,k} (\hat{\hv}^{m}_{l,k})^\H	
= \check{\Hm}^m_{[m]} \Am^m_{[m]} \lr\check{\Hm}^m_{[m]}\rr^\H
\end{align*}
where $\check{\Hm}^m_{[m]} \in \CC^{N_m \times K_{[m]}}$, with $K_{[m]}=\sum_{l\neq m} K_l$, is the aggregated matrix of the vectors $\check{\hv}^{m}_{l,k}\sim \CCC\NNN(0,\frac{1}{N_m}\Id_{N_m})\,, \forall\ l\neq m$ 
and 
\begin{align*}
\Am^m_{[m]} = \diag &\Big[ \alpha^m_1 \chi^m_{1,1}, \ldots, \alpha^m_1 \chi^m_{1,K_1}, \alpha^m_2 \chi^m_{2,1}, \ldots, \\
&\alpha^m_2 \chi^m_{2,K_2}, \cdots, \alpha^{m}_{m-1} \chi^m_{m-1,K_{m-1}}, \\
&\alpha^{m}_{m+1} \chi^m_{m+1,1}, \cdots, \alpha^{m}_{L} \chi^m_{L,K_L}\Big]
\end{align*}
i.e., a diagonal matrix with the terms pertaining to $\alpha^m_m$ removed.

One can directly apply \cite{SilversteinBai1995} or \cite{Couillet2011a}[Theorem~3.13, Eq~3.23] with $\Tm = \Am^m_{[m]}$ and $\Xm = (\check{\Hm}^m_{[m]})^\H$. Being careful with the notation ($\Xm \Tm \Xm^\H$ instead of $(\check{\Hm}^m_{[m]})^\H \Am^m_{[m]} \check{\Hm}^m_{[m]}$), we arrive at:
\begin{align*}
e_{m}		&= \frac{1}{N_m}\trace\lc \ls \check{\Hm}^m_{[m]} \Am^m_{[m]} (\check{\Hm}^m_{[m]})^\H + \beta_m \Id_{N_m} 
\rs^\mo  \rc 
\end{align*}
where				
\begin{align*}
e_{m} - \frac{1}{N_m} \ls \beta_m + 
\frac{1}{N_m} \sum_{l\neq m}^L \sum_k^{K_l} \frac{ \alpha^m_l \chi^m_{l,k} }{ 1+ \alpha^m_l \chi^m_{l,k} e_{m} }
\rs^\mo 		\xrightarrow[N \to \infty]{\mathrm{a.s.}} 0 \,.
\end{align*}
Here we used Remark~\ref{rem:iaRZF_simpleDEofe} and $\beta_m$ is given in~\eqref{eq:iaRZF_betaforsecondST}.

Combining the intermediate results, using Remark~\ref{rem:iaRZF_simpleDEofe} and the relationship $\trace\Am\lr\Am+x\Id\rr^\mo = \trace\Id-x\trace\lr\Am+x\Id\rr^\mo$ with $\Am=\Zm^m + \xi_m \Id_{N_m}$, we arrive at
\begin{align*}
\bar{X}_m = - \frac{1}{\alpha^m_m N_m} \sum_{j=1}^{K_m} \frac
{ \alpha^m_m\chi^m_{m,j} e_{m}^\prime }
{ (1+\alpha^m_m\chi^m_{m,j} e_{m})^2 }
\end{align*}
where $e_{m}^\prime$ is shorthand for $\partial/\partial \xi_m e_{m} $ and can found (by prolonged calculus) to be 
as stated in~\eqref{eq:iaRZF_ep}, which concludes this part of the proof.

\subsection{Signal Power Term}
\label{ssec:iaRZF_appAsilomarSig} 
The important part of finding the DE of the signal power term \eqref{eq:iaRZF_FDimSig} to find a DE of $(\hv^l_{l,k})^\H \Qm_l \hat{\hv}^l_{l,k}$, which will now first be done.
Before proceeding, we remind ourselves that our chosen model of the estimated channel \eqref{eq:iaRZF_estimatedchannel} entails the following relationships: $ \hv^l_{l,k} \independent \tilde{\hv}^l_{l,k}$, $\hat{\hv}^l_{l,k}  \not\independent \hv^l_{l,k}$, $\hat{\hv}^l_{l,k} \not\independent \tilde{\hv}^l_{l,k}$, $\Qm_{l[k]} \independent \hat{\hv}^l_{l,k}$, $\Qm_{l[k]} \independent \hv^l_{l,k}$. Also, formulations containing $\hat{\hv}^l_{l,k}$ can often be split into two terms comprising $\hv^l_{l,k}$ and $\tilde{\hv}^l_{l,k}$.
Hence, the application of Lemmas~\ref{lem:iaRZF_MILI}, \ref{lem:iaRZF_TrL}, \ref{lem:iaRZF_R1PL} and Corollary~\ref{cor:iaRZF_TrL0}, in the following is well justified. Employing \eqref{eq:iaRZF_SimpleEquationtoe} one sees
\begin{align*}
(\hv^l_{l,k})^\H \Qm_l \hat{\hv}^l_{l,k} 
- 
\frac{\sqrt{\chi^l_{l,k}}\sqrt{(1-(\tau^l_l)^2)} e_{(l)} }{1 + \alpha^l_l \chi^l_{l,k} e_{(l)} }
\xrightarrow[N \to \infty]{\mathrm{a.s.}} 0 \,.
\end{align*}
Finally, applying this result to the complete formulation \eqref{eq:iaRZF_FDimSig}, we arrive at the familiar term from Theorem~\ref{theo:iaRZF_DEofSINR}:
\begin{align*}
\overline{\mathrm{Sig}}^{(l)}_{l,k} =  \overline{\nu}_l (\chi^l_{l,k})^2 e_{(l)}^2 \lr 1-(\tau^l_l)^2 \rr (f^l_{l,k})^2 \,.
\end{align*}

\subsection{Preparation for Interference Terms}
\label{ssec:iaRZF_appAsilomarGenLem7} 
In this subsection we derive the deterministic equivalents of the two terms $(\hv^l_{l,k})^\H \Bm \Qm_l \hv^l_{l,k}$ and $(\hv^l_{l,k})^\H \Bm \Qm_l \tilde{\hv}^l_{l,k}$, where $\Bm \in \CC^{N_l\times N_l}$ has uniformly bounded spectral norm w.r.t.\ $N_l$ and is independent of $\hv^l_{l,k}$ and $\tilde{\hv}^l_{l,k}$.
The following approach is based on and slightly generalizes \cite[Lemma~7]{Wagner2012a}.
First, 
it is helpful to realize an implication of our resolvent notation (Definition~\ref{def:iaRZF_DefResolvents}) and channel estimation model \eqref{eq:iaRZF_estimatedchannel}:
\begin{align*}
\Qm_{a}^\mo - \Qm_{a[bc]}^\mo &=  c_0 \hv^{a}_{b,c} (\hv^{a}_{b,c})^\H + c_2 \hv^{a}_{b,c} (\tilde{\hv}^{a}_{b,c})^\H  + \\
& \qquad c_2 \tilde{\hv}^{a}_{b,c} (\hv^{a}_{b,c})^\H + c_1 \tilde{\hv}^{a}_{b,c} (\tilde{\hv}^{a}_{b,c})^\H  \numberthis \label{eq:iaRZF_Lem7SebResSum}
\end{align*}
where $c_0 = \alpha^a_b \chi^a_{b,c} \lr 1 - (\tau^a_b)^2 \rr$, $c_1 = \alpha^a_b \chi^a_{b,c} (\tau^a_b)^2$ and $c_2 = \alpha^a_b \chi^a_{b,c} \sqrt{(1-(\tau^a_b)^2)} \tau^a_b$. 
We omitted designating the dependencies of $c$ on $a$ and $b$, as this is always clear from the context.
To ease the exposition, we also introduce the following abbreviations
\begin{align*}
Y_1 &\eqdef (\tilde{\hv}^l_{l,k})^\H \Qm_{l[k]} \hv^l_{l,k} & Y_4 &\eqdef (\hv^l_{l,k})^\H \Bm \Qm_{l[k]} \hv^l_{l,k} \\
Y_2 &\eqdef (\hv^l_{l,k})^\H \Qm_{l[k]} \tilde{\hv}^l_{l,k} & Y_5 &\eqdef (\tilde{\hv}^l_{l,k})^\H \Qm_{l[k]} \tilde{\hv}^l_{l,k} \\
Y_3 &\eqdef (\hv^l_{l,k})^\H \Bm \Qm_{l[k]} \tilde{\hv}^l_{l,k}  & Y_6 &\eqdef (\hv^l_{l,k})^\H \Qm_{l[k]} \hv^l_{l,k} \,.
\end{align*}
Finally, we begin with the term $(\hv^l_{l,k})^\H \Bm \Qm_l \tilde{\hv}^l_{l,k}$:
\begin{align*}
&(\hv^l_{l,k})^\H \Bm \Qm_l \tilde{\hv}^l_{l,k} - (\hv^l_{l,k})^\H \Bm \Qm_{l[k]} \tilde{\hv}^l_{l,k} \eqtextt{\eqref{eq:iaRZF_ResolventId}} \\ 
& \qquad 
- (\hv^l_{l,k})^\H \Bm \Qm_l \lr \Qm_l^\mo - \Qm_{l[k]}^\mo \rr \Qm_{l[k]} \tilde{\hv}^l_{l,k}
\end{align*}
and, using \eqref{eq:iaRZF_Lem7SebResSum}, we find 
\begin{align}
&(\hv^l_{l,k})^\H \Bm \Qm_l \tilde{\hv}^l_{l,k} =
\frac{
	Y_3
	- (\hv^l_{l,k})^\H \Bm \Qm_l \hv^l_{l,k}  
	\lr
	c_0 Y_2 + c_2 Y_5
	\rr
}{
1 + c_2 Y_2 + c_1 Y_5
}\, . \label{eq:iaRZF_Lem7Seb2} 
\end{align}
Similarly, for the term $(\hv^l_{l,k})^\H \Bm \Qm_l \hv^l_{l,k}$ we arrive at
\begin{align}
&(\hv^l_{l,k})^\H \Bm \Qm_l \hv^l_{l,k} 
\lr 
1 + c_0 Y_6 + c_2 Y_1
\rr  \nonumber \\ 
& \qquad = 
Y_4  - (\hv^l_{l,k})^\H \Bm \Qm_l \tilde{\hv}^l_{l,k} 
\lr
c_2 Y_5 + c_1 Y_1
\rr \, . 	
\label{eq:iaRZF_Lem7Seb1p} 
\end{align}
Now, applying \eqref{eq:iaRZF_Lem7Seb2} to \eqref{eq:iaRZF_Lem7Seb1p}, one arrives at 
\begin{align*}
&(\hv^l_{l,k})^\H \Bm \Qm_l \hv^l_{l,k} \\
& \times \ls
\lr 
1 + c_0 Y_6 + c_2 Y_1 
\rr 
- 
\frac{
	\lr c_0 Y_2 + c_2 Y_5 \rr\lr c_2 Y_6 + c_1 Y_1 \rr
}{
1 + c_2 Y_2 + c_1 Y_5
}
\rs \\
&\qquad = 
Y_4
- \frac{
	(\hv^l_{l,k})^\H \Bm \Qm_l \tilde{\hv}^l_{l,k} 
	\lr
	c_2 Y_6 + c_1 Y_1 
	\rr 
}{
1 + c_2 Y_2 + c_1 Y_5
}
\,. \numberthis \label{eq:iaRZF_Lem7SebBeforeFirstAsy}
\end{align*}
Similar to Appendix~\ref{ssec:iaRZF_appAsilomarSig}, we notice that $Y_1$, $Y_2$ and $Y_3$, converge almost surely to $0$ in the large system limit:
\begin{align*}
Y_1, Y_2, Y_3 \xrightarrow[N \to \infty]{\mathrm{a.s.}} 0 \,.
\end{align*}
We also foresee that
\begin{align*}
Y_4 - u^\prime \xrightarrow[N \to \infty]{\mathrm{a.s.}} 0\,, \quad
Y_5 - u_1 \xrightarrow[N \to \infty]{\mathrm{a.s.}} 0\,, \quad
Y_6 - u_2 \xrightarrow[N \to \infty]{\mathrm{a.s.}} 0 
\end{align*}
where the values for $u^\prime\,, u_1$ and $u_2$ are not yet of concern.
Thus, \eqref{eq:iaRZF_Lem7SebBeforeFirstAsy} finally leads to
\begin{align*}
(\hv^l_{l,k})^\H \Bm \Qm_l \hv^l_{l,k} 
\ls
\lr 
1 + c_0 u_2 
\rr 
-
\frac{
	\lr c_2 u_1 \rr\lr c_2 u_2 \rr
}{
1 + c_1 u_1
}
\rs - u^\prime
\xrightarrow[N \to \infty]{\mathrm{a.s.}} 0
\end{align*}
and we finally find the expression we were looking for
\begin{align*}
(\hv^l_{l,k})^\H \Bm \Qm_l \hv^l_{l,k} - \frac{u^\prime\lr 1 + c_1 u_1 \rr}{1+c_1u_1+c_0u_2+\lr c_0c_1-c_2^2\rr u_1u_2} \xrightarrow[N \to \infty]{\mathrm{a.s.}} 0
\,. \numberthis \label{eq:iaRZF_Lem7SebAfterAsyTerm1}
\end{align*}

In order to find the second original term ($(\hv^l_{l,k})^\H \Bm \Qm_l \tilde{\hv}^l_{l,k}$), we reform and plug \eqref{eq:iaRZF_Lem7Seb1p} into \eqref{eq:iaRZF_Lem7Seb2} 
and follow analogously the path we took to arrive at \eqref{eq:iaRZF_Lem7SebAfterAsyTerm1}. We finally find
\begin{align*}
(\hv^l_{l,k})^\H \Bm \Qm_l \tilde{\hv}^l_{l,k} - \frac{-c_2u_1u^\prime}{1+c_1u_1+c_0u_2+\lr c_0c_1-c_2^2\rr u_1u_2} \xrightarrow[N \to \infty]{\mathrm{a.s.}} 0	\,. \numberthis \label{eq:iaRZF_Lem7SebAfterAsyTerm2}
\end{align*}

\subsection{Interference Power Terms}
\label{ssec:iaRZF_appAsilomarInt} 
Having obtained the preparation results in Appendix~\ref{ssec:iaRZF_appAsilomarGenLem7} we can now continue to find the DEs for different parts of the interference power term. From \eqref{eq:iaRZF_FDimInt} we arrive at
\begin{equation}
\begin{aligned}
\mathrm{Int}^{(l)}_{l,k} 			
&= \sum_{m\neq l} \nu_m \chi^m_{l,k} \underbrace{(\hv^m_{l,k})^\H \Qm_m \hat{\Hm}^m_m  (\hat{\Hm}^m_m)^\H \Qm_m \hv^m_{l,k}}_{\text{Part~A$_m$}} \\
& \qquad + \nu_l \chi^l_{l,k} \underbrace{(\hv^l_{l,k})^\H \Qm_l \hat{\Hm}^l_{l[k]} (\hat{\Hm}^l_{l[k]})^\H \Qm_l  \hv^l_{l,k}}_{\text{Part~B}} \,. 
\end{aligned}
\label{eq:iaRZF_IntABsplit}
\end{equation}
We start by treating \eqref{eq:iaRZF_IntABsplit}~Part~B first. Employing the relationships $\Am\Bm\Dm = \Am\Cm\Dm + \Am(\Bm-\Cm)\Dm$ and \eqref{eq:iaRZF_ResolventId} one finds 
\begin{align*}
\text{Part~B}  &=
(\hv^l_{l,k})^\H \Qm_{l[k]} \hat{\Hm}^l_{l[k]} (\hat{\Hm}^l_{l[k]})^\H \Qm_l \hat{\Hm}^l_{l[k]} \hv^l_{l,k} \\
& - (\hv^l_{l,k})^\H \Qm_l \ls \Qm_l^\mo - \Qm_{l[k]}^\mo \rs \Qm_{l[k]} \hat{\Hm}^l_{l[k]} (\hat{\Hm}^l_{l[k]})^\H \Qm_l  \hv^l_{l,k} \,.
\end{align*}
Using the relationship~\eqref{eq:iaRZF_Lem7SebResSum} pertaining to $\ls \Qm_l^\mo - \Qm_{l[k]}^\mo \rs$, we can split Part~B as
\begin{align*}
\text{Part~B}
& = X_1 - c_0 X_3 X_1 - c_2 X_3 X_2 - c_2 X_4 X_1 - c_1 X_4 X_2 \,.
\end{align*}
Where we have found and abbreviated the $4$ quadratic forms,
\begin{align*}
X_1 &= (\hv^l_{l,k})^\H \Qm_{l[k]} \hat{\Hm}^l_{l[k]} (\hat{\Hm}^l_{l[k]})^\H \Qm_l \hv^l_{l,k} \\
X_2 &= (\tilde{\hv}^l_{l,k})^\H \Qm_{l[k]} \hat{\Hm}^l_{l[k]} (\hat{\Hm}^l_{l[k]})^\H \Qm_l  \hv^l_{l,k} \\
X_3 &= (\hv^l_{l,k})^\H \Qm_l \hv^l_{l,k} \\
X_4 &= (\hv^l_{l,k})^\H \Qm_l \tilde{\hv}^l_{l,k} \,.
\end{align*}

To find the deterministic equivalents for $X_1$ and $X_2$, we can use \eqref{eq:iaRZF_Lem7SebAfterAsyTerm1} and \eqref{eq:iaRZF_Lem7SebAfterAsyTerm2}, respectively, where $\Bm = \Qm_{l[k]} \hat{\Hm}^l_{l[k]} (\hat{\Hm}^l_{l[k]})^\H$.
The respective variables $u_1, \, u_2$ and $u^\prime$ for this choice of $\Bm$ are found (using the same standard techniques as in Appendix~\ref{ssec:iaRZF_appAsilomarSig}) to be
\begin{align*}
u_1 = (\tilde{\hv}^l_{l,k})^\H \Qm_{l[k]} \tilde{\hv}^l_{l,k} \qquad
\Rightarrow \quad &u_1 - e_{l} \xrightarrow[N \to \infty]{\mathrm{a.s.}} 0\,.
\end{align*}
Analogously,
\begin{align*}	
u_1 - e_{(l)} \xrightarrow[N \to \infty]{\mathrm{a.s.}} 0\,.
\end{align*}
Hence, we see that $u_1$ and $u_2$ converge to the same value and we will abbreviate them henceforth as $u$. For the still missing term $u^\prime$ we arrive at
\begin{align*}	
u^\prime &= (\hv^l_{l,k})^\H \Qm_{l[k]} \hat{\Hm}^l_{l[k]} (\hat{\Hm}^l_{l[k]})^\H \Qm_{l[k]} \hv^l_{l,k} \\
& \Rightarrow u^\prime - g_{l} \xrightarrow[N \to \infty]{\mathrm{a.s.}} 0
\end{align*}
where the last step makes have use of the results in Appendix~\ref{ssec:iaRZF_appAsilomarPwrNorm}.
Also, we remind ourselves that we have $c_0 = \alpha^l_l \chi^l_{l,k} \lr 1 - (\tau^l_l)^2 \rr$, $c_1 = \alpha^l_l \chi^l_{l,k} (\tau^l_l)^2$ and $c_2 = \alpha^l_l \chi^l_{l,k} \sqrt{(1-(\tau^l_l)^2)} \tau^l_l$, hence $c_0+c_1 = \alpha^l_l \chi^l_{l,k}$ and $c_0c_1-c_2^2 = 0$.
So, finally, we have
\begin{align*}
X_1 - \frac{u^\prime\lr 1 + c_1 u \rr}{1+ \lr c_1+c_0 \rr u} \xrightarrow[N \to \infty]{\mathrm{a.s.}} 0 
\end{align*}
and similarly 
\begin{align*}
\text{ and }
X_2 - \frac{-c_2uu^\prime}{1+ \lr c_1+c_0 \rr u} \xrightarrow[N \to \infty]{\mathrm{a.s.}} 0 \,.
\end{align*}

To find the DEs for $X_3$ and $X_4$, we can again use \eqref{eq:iaRZF_Lem7SebAfterAsyTerm1} and \eqref{eq:iaRZF_Lem7SebAfterAsyTerm2}, respectively. This time $\Bm = \Id$ and hence the variables simplify to 
$
u^\prime = u_1 = u_2 \eqdef u
$, 
where
$
u - e_{l} \xrightarrow[N \to \infty]{\mathrm{a.s.}} 0\,.
$
Thus,
\begin{align*}
X_3 - \frac{u \lr 1 + c_1 u \rr}{1+ \lr c_1+c_0 \rr u} & \xrightarrow[N \to \infty]{\mathrm{a.s.}} 0 \\
X_4 - \frac{-c_2 u^2}{1+ \lr c_1+c_0 \rr u} & \xrightarrow[N \to \infty]{\mathrm{a.s.}} 0 \,.
\end{align*}

Combining all results after further simplifications, we can express the DE of Part~B, i.e., $\overline{\text{Part~B}}$, as 
\begin{align*}
\overline{\text{Part~B}} = g_{l} \frac{1-(\tau^l_l)^2}{\lr 1+\alpha^l_l\chi^l_{l,k}e_{l} \rr^2} + g_{l} (\tau^l_l)^2 \,.
\end{align*}

The next step is to derive the DE of \eqref{eq:iaRZF_IntABsplit}~Part~A$_m$, i.e., $\overline{\text{Part~A}}_m$.
Fortunately, the sum obliges $m\neq l$ and, thus, the same derivation like for Part~B applies.
Hence, we arrive at
\begin{align*}
\overline{\text{Part~A}}_m = g_{m} \frac{1-(\tau^m_l)^2}{\lr 1+\alpha^m_l\chi^m_{l,k}e_{m} \rr^2} + g_{m} (\tau^m_l)^2 \,.
\end{align*}

Combing Part~B and the sum of Part~A$_m$ with our original expression of the interference power, we  arrive at the familiar expression from Theorem~\ref{theo:iaRZF_DEofSINR}.


\bibliographystyle{IEEEtran}
\bibliography{IEEEabrv,JabrefDatabaseThesisFinal}

\end{document}